\documentclass[10pt, twocolumn, comsoc]{IEEEtran}

\IEEEoverridecommandlockouts
\usepackage{cite}
\usepackage{amsmath,amssymb,amsfonts}
\DeclareMathOperator*{\argmax}{arg\,max}

\usepackage{algorithmic}
\usepackage{graphicx}
\usepackage{textcomp}
\usepackage{xcolor}
\usepackage{amsthm}
\usepackage{verbatim}
\usepackage{multirow}
\newtheorem{theorem}{Theorem}

\usepackage{caption}
\usepackage{subcaption}

\newtheorem{lemma}{Lemma}

\newtheorem{corollary}{Corollary}

\def\d{{\rm d}}

\usepackage{suffix}
\usepackage{mathtools}

\DeclarePairedDelimiterX\MeijerM[3]{\lparen}{\rparen}%
{\begin{smallmatrix}#1 \\ #2\end{smallmatrix}\delimsize\vert\,#3}

\newcommand\MeijerG[8][]{%
  G^{\,#2,#3}_{#4,#5}\MeijerM[#1]{#6}{#7}{#8}}

\WithSuffix\newcommand\MeijerG*[7]{%
  G^{\,#1,#2}_{#3,#4}\MeijerM*{#5}{#6}{#7}}

\def\BibTeX{{\rm B\kern-.05em{\sc i\kern-.025em b}\kern-.08em
    T\kern-.1667em\lower.7ex\hbox{E}\kern-.125emX}}
\begin{document}

\title{ A Tractable Approach to Coverage Analysis in Downlink Satellite Networks
}


\author{\IEEEauthorblockN{ Jeonghun~Park},
{\it  Member,~IEEE},
\and
\IEEEauthorblockN{Jinseok~Choi},
{\it Member,~IEEE}\\
 \and 
 \IEEEauthorblockN{Namyoon~Lee},
 {\it Senior Member,~IEEE}
\thanks{ J. Park is with the School of Electronics Engineering, Kyungpook National University, South Korea (e-mail: {\texttt{jeonghun.park@knu.ac.kr}}). J. Choi is with Department of Electrical Engineering, Ulsan National Institute of Science and Technology, South Korea (e-mail: {\texttt{jinseokchoi@unist.ac.kr}}). N. Lee is with the School of Electrical Engineering, Korea University, South Korea (e-mail: {\texttt{namyoon@korea.ac.kr}})} \thanks{This work was supported by the National Research Foundation of Korea (NRF) grant funded by the Korea government (MSIT) (No. 2020R1C1C1013381) and Institute of Information and communications Technology Planning and Evaluation(IITP) grant funded by the Korea government(MSIT) (No.2021-0-00260,Research on LEO Inter-Satellite Links).}}
 
\maketitle

\begin{abstract} 
Satellite networks are promising to provide ubiquitous and high-capacity global wireless connectivity. Traditionally, satellite networks are modeled by placing satellites on a grid of multiple circular orbit geometries. Such a network model, however, requires intricate system-level simulations to evaluate coverage performance, and analytical understanding of the satellite network is limited. Continuing the success of stochastic geometry in a tractable analysis for terrestrial networks, in this paper, we develop novel models that are tractable for the coverage analysis of satellite networks using stochastic geometry. By modeling the locations of satellites and users using Poisson point processes on the surfaces of concentric spheres, we characterize analytical expressions for the coverage probability of a typical downlink user as a function of relevant parameters, including path-loss exponent, satellite height, density, and Nakagami fading parameter. Then, we also derive a tight lower bound of the coverage probability in tractable expression while keeping full generality. Leveraging the derived expression, we identify the optimal density of satellites in terms of the height and the path-loss exponent. Our key finding is that the optimal average number of satellites decreases logarithmically with the satellite height to maximize the coverage performance. Simulation results verify the exactness of the derived expressions. 

 \end{abstract}




\section{Introduction}
Universal coverage is a persistent trend in wireless communications. Satellite networks integrated with terrestrial cellular networks can be promising solutions to provide seamless and ubiquitous connectivity services \cite{liu:commmag:21, giordani:network:21}. While classical satellite networks using geosynchronous equatorial orbit (GEO) are effective at providing stationary coverage to a specific area, the attention of researchers is recently shifting to satellite networks employing the low Earth orbit (LEO) or very LEO (VLEO) mega-satellite constellations. Unlike GEO satellite networks, LEO or VLEO satellite networks can achieve higher data rates with much lower delays at the cost of deploying more dense satellites to attain global coverage performance. For instance, various satellite network companies have recently been deploying about a few thousand VLEO and LEO satellites below 1000 km elevations to provide universal internet broadband services on Earth.

The characterization of the coverage and rate performance for VLEO satellite networks is of importance because of ultra-expensive costs for deploying mega VLEO satellites. Satellite networks are conventionally modeled by placing satellites on a grid of multiple circular orbit geometries, e.g., the Walker constellation. This model, however, is not very analytically tractable to characterize coverage and rate performance. For this reason, intricate system-level simulations are required to evaluate such performance by numerically averaging out the many sources of randomness, including satellites' locations and channel fading processes. 
Motivated by this, an analytical approach is indispensable as a complementary element to obtain insights on network design guidance.


 Stochastic geometry is a mathematical tool that characterizes the spatial distributions of base stations' and users' locations in wireless networks \cite{gilbert:josiam:61,baccelli:book:09,haenggi:tit:08}.  Significant progress has been made in recent years on characterizing wireless networks' coverage and rate performances using stochastic geometry. By modeling the placements of base stations and users according to Poisson point processes (PPPs), analytically tractable expressions for coverage and rates, which produce useful insights on network design guidance, have been found for ad-hoc \cite{baccelli:book:09,baccelli:tit:06,haenggi:tit:08,baccelli:jsac:09,haenggi:jsac:09}, cellular \cite{andrews:tcom:11, dhillon:jsac:12, lee:jsac:15}, multi-antenna \cite{lee:twc:15,huang:tit:13,park:twc:16, lee:tit:16}, mmWave \cite{bai:commmag:14,park:tccn:18,renzo:twc:15,atzeni:twc:18}, and UAVs \cite{chetlur:tcom:17,banagar:twc:20} networks. Continuing in the same spirit, in this work, we develop a tractable model for satellite downlink networks and characterize the coverage probability to illuminate the design principle of satellite downlink networks in terms of relevant network parameters. 
 

 
  
 \subsection{Related Work}



  
  Satellite network models based on stochastic geometry are recently obtaining momentum compared to deterministic grid models such as Walker constellations \cite{ganz:tcom:94,vatalaro:jsac:95,mokhtar:wcl:20,seyedi:commlett:12} because of their superiority in analytical tractability.  The most popular approach is to model the spatial distribution of satellites' locations according to a homogenous binomial point process (BPP) on a spherical shell. Assuming that a fixed number of satellites are uniformly distributed on the surface of a sphere, the analytical expressions for the coverage probability and the rate are derived in \cite{okati:tcom:20}. By modeling the locations of satellites according to multiple BPPs on the surfaces of concentric spheres,  \cite{talgat:commlett:20} characterizes the nearest-neighbor distance distributions depending on two different receivers' locations. Utilizing the derived distance distributions in \cite{talgat:commlett:20}, \cite{talgat:commlett:21} establishes an expression for the downlink coverage probability under a scenario where satellites play as relays between the users and the LEO satellites \cite{talgat:commlett:21}.  
  
%

  PPPs have also been used to model the spatial distribution of satellites'  placements in a portion on the surface of a sphere, i.e., a finite space.  Unlike the case of a BPP,  the number of satellites uniformly distributed over the spherical cap is modeled to follow a Poisson distribution. In \cite{hourani:wcl:21, hourani:wcl:21_2}, the downlink coverage probability of a satellite network is derived using the approximation of the contact angle distribution.  In \cite{okati:tcom:22},  the coverage and rate are analyzed for a noise-limited LEO satellite network by modeling satellites' locations according to a nonhomogeneous PPP with a non-uniform satellite intensity depending on latitudes.  These prior works in \cite{okati:tcom:20,okati:tcom:22,hourani:wcl:21, hourani:wcl:21_2,talgat:commlett:20,talgat:commlett:21} extend the realm of stochastic geometry into modeling satellite networks on the spherical cap. Notwithstanding the progress, the derived expressions for the coverage and rate are not very analytically tractable compared to the astonishingly simple expression on the coverage probability obtained in cellular networks that assumed an infinite two-dimensional space \cite{baccelli:book:09,baccelli:tit:06,haenggi:tit:08,baccelli:jsac:09,haenggi:jsac:09, andrews:tcom:11,dhillon:jsac:12,renzo:tcom:13,lee:jsac:15}. 
  
Applying a BPP or a PPP to model a wireless network in a finite space makes the performance analysis more challenging than using a PPP network in an infinite space  \cite{afshang:twc:17,haenggi:tit:05}. The difficulty primarily arises because the signal-to-interference-plus-noise ratio (SINR) depends on a receiver's location, i.e., a typical node analysis does not hold by boundary effects. In addition, the nearest distance distribution, which is a crucial instrument for the coverage probability, is less tractable than the counterpart of a PPP in an infinite space \cite{afshang:twc:17,haenggi:tit:05}. To be specific, when $N$ satellites are independent and identically distributed (IID) in a finite area (i.e., the surface of a sphere), the cumulative distribution function (CDF) of the nearest distribution is obtained by taking the $N$th power of the complementary CDF (CCDF) of the individual distance distribution as derived in \cite{okati:tcom:20,talgat:commlett:20}. As a result, this nearest distribution is unwieldy to compute the Laplace transform of the aggregated interference power. It can also be further complicated when marginalizing the effect of the number of satellites according to the Poisson distribution, making the closest distance distribution much less analytically tractable.  

\subsection{Contributions}
This paper puts forth a novel framework for evaluating satellite downlink networks in a finite space. In particular, we model the locations of satellites and users as independent and homogenous PPPs on the surfaces of spheres with radius $R_{\sf S}$ and $R_{\sf E}$, where $R_{\sf E}$ denotes the radius of the Earth.  Thanks to Slivnyak's theorem \cite{haenggi:tit:05}, we focus on a typical receiver located at (0,0, $R_{\sf E}$) on the surface of the Earth.  Then, a typical spherical cap is defined by a portion of the surface on the sphere in the field of view at the typical receiver's location.  The main contributions are summarized as follows:
\begin{itemize}
	\item As a stepping stone towards the computation of the coverage probability, we first derive a new nearest distance distribution from a typical receiver's location to the satellite when the satellites are distributed according to a PPP in a finite area, the surface of a sphere. Unlike the previous approaches in \cite{okati:tcom:20,okati:tcom:22,hourani:wcl:21, hourani:wcl:21_2,talgat:commlett:20,talgat:commlett:21}, we compute the nearest distance distribution conditioned that at least one satellite exists in the spherical cap seen by the typical receiver. It turns out that the derived conditional nearest distance distribution is a truncated \textit{Rayleigh distribution}, which is the time-honored nearest distance distribution for a PPP in an infinite space.  
 
	
	\item Harnessing a tractable form of the nearest distance distribution, we first establish an \textit{exact} expression for the coverage probability of a typical receiver when fading processes follow the Nakagami-$m$ distribution.  To accomplish this, we begin by deriving the coverage probability conditioned that a positive number of satellites are placed on the typical spherical cap in terms of system parameters, chiefly path-loss exponent, density, and the altitude of satellites. It is then marginalized to obtain the exact coverage probability. 
	\item We further derive a lower and upper bound of the coverage probability in tractable forms, while keeping full generality. We further derive a tight lower bound of the coverage probability in a tractable form, which illuminates the features of the system parameters in the coverage probability. 
	{\color{black}{The analytical coverage probabilities are verified by comparing to the coverage probability drawn by using actual Starlink satellite constellations data sets.  }}
	\item  We identify the optimal density of the satellites for the typical user as a function of the satellite's height and the path-loss exponent. Our finding is that the optimal average number of satellites scales down as the altitude of satellites becomes higher. 

\end{itemize}

The remainder of the paper is organized as follows. Section II describes the proposed network and channel models with the performance metrics for the downlink coverage probability. In Section III, the analytical expressions for the coverage probabilities are derived. Section IV characterizes the optimal density of the satellites according to system parameters.  Section V provides simulations results to validate our analysis.  Section VI concludes the paper.

 \section{System Model}
 In this section, we explain the proposed network model and the performance metrics for the analysis of downlink satellite networks. 

\subsection{Network Model }

{\bf Surface of a sphere:} The surface of a sphere in $\mathbb{R}^3$ with center at the origin ${\bf 0}\in \mathbb{R}^3$ and fixed radius $R_{\sf S}$ is defined as
\begin{align}
	\mathbb{S}_{R_{\sf S}}^2=\{{\bf x}\in \mathbb{R}^3: \|{\bf x}\|_2=R_{\sf S}\}.
\end{align}  
A point vector ${\bf x}\in \mathbb{S}_{R_{\sf S}}^2$ can be represented as a pair of elevation and azimuth angles $0\leq \theta \leq 2\pi$ and  $0\leq \phi \leq 2\pi$ in a polar coordinate.

\begin{figure} 
    \centering 
    \includegraphics[width=0.85\columnwidth]{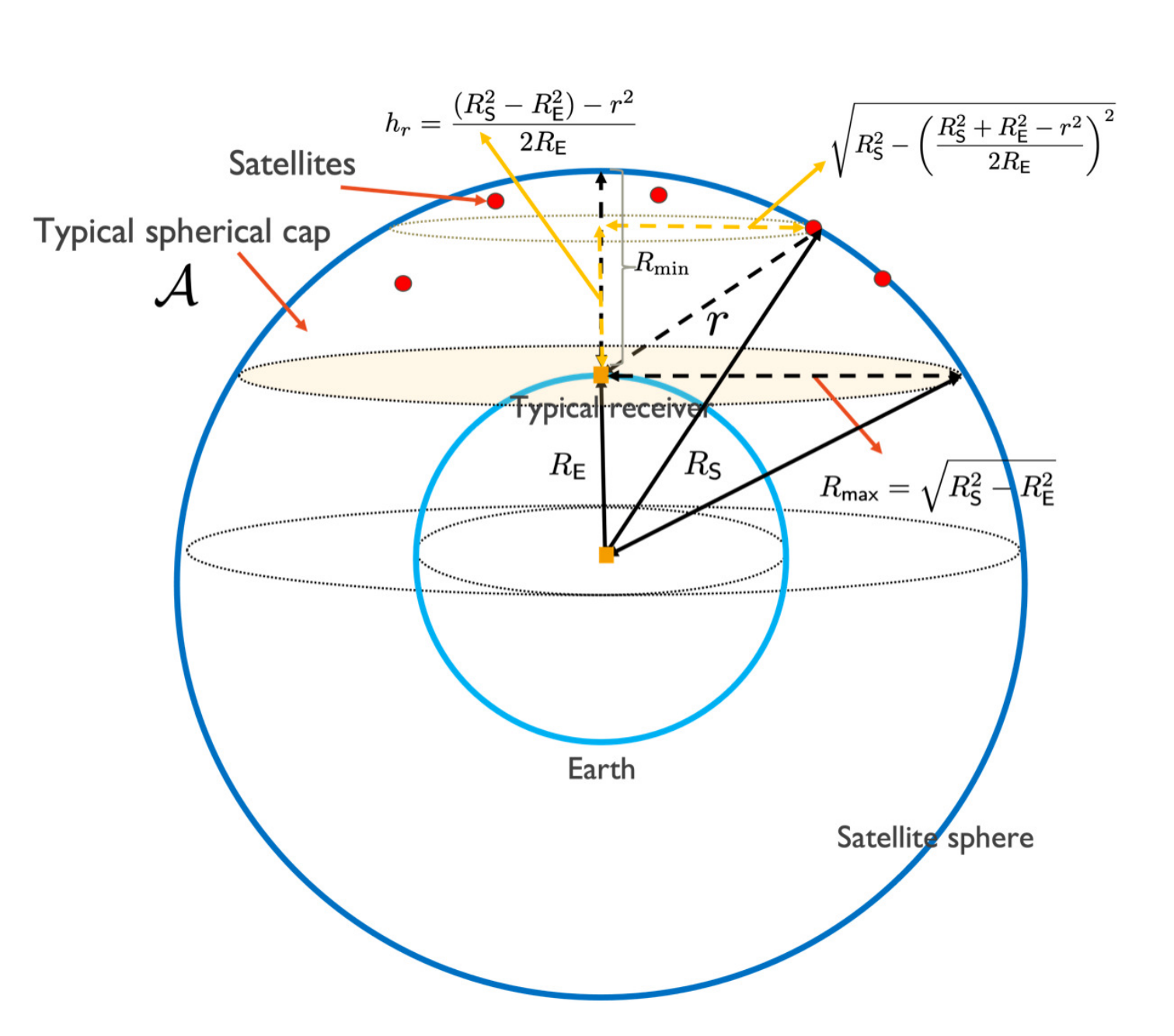}
    \caption{Geometry of satellite networks. Satellites are assumed to be placed on the surface of a sphere with radius $R_{\sf S}$. A typical receiver is located at $(0,0,R_{\sf E})$. The distance from the typical receiver to a satellite with distance $r$ is calculated in terms of geometry parameters of $R_{\sf S}$ and $R_{\sf E}$.}\label{fig:geo}
\end{figure}

{\bf  Poisson point process on the surface of a sphere:} Let $\Phi=\{{\bf x}_1,\ldots, {\bf x}_N\}$ be a point process consisted of a finite number elements on the surface of a sphere $\mathbb{S}_{R_{\sf S}}^2$.
$\Phi$ is said to a homogenous spherical Poisson point process (SPPP), provided that the number of points on $\mathbb{S}_{R_{\sf S}}^2$, $N=\Phi(\mathbb{S}_{R_{\sf S}}^2)$, follows Poisson random variable with mean $\lambda|\mathbb{S}_{R_{\sf S}}^2|=4\pi R_{\sf S}^2\lambda$, namely,
\begin{align}
	\mathbb{P}\left( N=n \right) =\exp\left(-4\pi R_{\sf S}^2\lambda\right)\frac{\left(4\pi R_{\sf S}^2\lambda\right)^n}{n!}, \label{eq:poisson}
\end{align} where $|\mathbb{S}_{R_{\sf S}}^2|=4\pi R_{\sf S}^2$ is the surface area of the sphere. For given $N$, the $\{{\bf x}_1,\ldots, {\bf x}_N\}$ forms a BPP, in which ${\bf x}_i$ for $i\in [N]$ is independent and uniformly distributed on the surface of the sphere.

{\bf Spatial distribution of satellites and users: }
 We assume that the satellites are located on the surface of the sphere with radius $R_{\sf S}$.  The locations of them are distributed according to a homogeneous SPPP with density $\lambda$, i.e., $\Phi=\{{\bf x}_1,\ldots, {\bf x}_{N}\}$, where $N$ follows a Poisson distribution with mean $4\lambda\pi R_{\sf S}^2 $.  Each satellite uses transmit power $P$.  We further consider a collection of downlink users. Let $\mathbb{S}_{R_{\sf E}}^2$ be the surface of the Earth with the radius $R_{\sf E}~ (<R_{\sf S})$. The locations of users on $\mathbb{S}_{R_{\sf E}}^2$ are assumed to be distributed according to a homogeneous SPPP with density $\lambda_{\sf U}$, i.e.,  $\Phi_{\sf U} = \{ {\bf u}_{1},  \ldots, {\bf u}_{M}\}$, where $M$ follows a Poisson distribution with mean $4\lambda_{\sf U}\pi R_{\sf E}^2 $.  We assume that $\Phi_{\sf U}$ is independent of underlying satellite placement processes $\Phi$.

{\bf Typical spherical cap:}
Since we assume that both $\Phi$ and $\Phi_{\sf U}$ are distributed according to the homogeneous SPPPs on $\mathbb{S}_{R_{\sf S}}^2$ and $\mathbb{S}_{R_{\sf E}}^2$, the statistical distribution of $\Phi$ with respect to any point in $\Phi_{\sf U}$ is invariant under rotation in $\mathbb{R}^3$. Thanks to the Slivnyak's theorem \cite{baccelli:book:09}, without loss of generality, we consider a typical user located at $\mathbb{S}_{R_{\sf E}}^2$ in $(0,0,R_{\sf E})$.  As depicted in Fig. \ref{fig:geo}, we define a typical spherical cap $\mathcal{A}\subset  \mathbb{S}_{R_{\sf S}}^2$ in the field of view at the typical receiver's location. In other words, the typical spherical cap is the portion of sphere $\mathbb{S}_{R_{\sf S}}^2$ cut off by a tangent plane to sphere $\mathbb{S}_{R_{\sf E}}^2$ (the Earth) centered at $(0,0,R_{\sf E})$.  Applying Archimedes' Hat-Box Theorem \cite{cundy:math:89}, the area of the typical spherical cap is given by\begin{align}
	|\mathcal{A}|=2\pi (R_{\sf S}-R_{\sf E})R_{\sf S}.\label{eq:area}
\end{align}
Also, we define a spherical cap that contains the points with smaller distance than $r$ from the typical receiver's location as 
\begin{align}
	\mathcal{A}_r=\{{\bf x}\in \mathbb{S}_{R_{\sf S}}^2: \|{\bf x}-(0,0,R_{\sf E})\|_2\leq r \}\subset \mathcal{A}.
\end{align} 
Without loss of generality, we focus on a downlink user performance in the typical spherical cap. 
{\color{black}{
We note that only the satellites located on the typical spherical cap $\mathcal{A}$ are visible to the typical user. Spatial distribution on these visible satellites is a PPP, whose density is $\lambda |\mathcal{A}|$. 
}}



\subsection{Path-Loss and Channel Models }
The propagation through the wireless channel is modeled by the combination of path-loss attenuation and small-scale fading.  We adopt the classical path-loss model, which is dependent on the distance from satellite $i\in[N]$ to the typical receiver, defined as 
\begin{align}
	  \|{\bf x}_{i}-{\bf u}_1\|^{-\alpha}.
\end{align}
This time-honored frequency-independent path-loss model may not be suitable for accurately capturing the effects of atmospheric absorption and rain fading. Notwithstanding such limitations, we shall focus on this model due to the analytical tractability and the path-loss exponent, effective beamforming gains, and small-scale fading parameters can be turned to incorporate these compound effects at a certain level of accuracy. 

We also model transmit and receive beamforming gains from  satellites to the typical receiver. Let $G_i$ be the effective antenna gain for the signal path from satellite $i$ to the receiver.  To make analysis tractable, we assume that the receive beam is perfectly aligned to the antenna boresight of the nearest serving satellite, while it is misaligned to those of interfering satellites. Then, the effective antenna gain is given by:
  \begin{align}
  	G_i =
\begin{cases}
G_1^{\rm t}G_1^{\rm r} \frac{c^2}{(4\pi f_c)^2},~~i=1,\\
G_i^{\rm t}G_1^{\rm r}\frac{c^2}{(4\pi f_c)^2},~~i\neq 1.
\end{cases}
  \end{align}
where $G_i^{\rm t}$ and $G_1^{\rm r}$ are the transmit antenna gain of satellite $i$ and the receive antenna gain of the typical receiver, $f_c$ denotes carrier frequency, and $c$ is the speed of light. It is worth mentioning that this antenna gain model arises from a two-lobe approximation of the antenna radiation pattern as in \cite{renzo:twc:15,bai:commmag:14}.  This two-lobe approximation can also be accurate when adopting Dolph-Chebyshev beamforming weights with uniform linear arrays \cite{koretz:tsp:09}.  Although simplified, this model has been widely used for its tractability while capturing primary features, including the directivity gain and the half-power beamwidth. We assume that $G_i$ is a constant for $i\in [N]/\{1\}$. Modeling the randomness on $G_i$ capturing the beam misalignment effect with more sophisticated distributions can be interesting future work.

We model the randomness of small-scale channel fading process using the Nakagami-$m$ distribution to suitably capture the line-of-sight (LOS) effects. Let $H_i$ be the fading from satellite $i$ to the typical receiver. Assuming $\mathbb{E}[H_i]=1$, the probability density function (PDF) of {\color{black}$\sqrt{H_i}$} is given by \cite{giunta:wcl:18}:
\begin{align} \label{eq:f_H}
	 f_{\color{black}\sqrt{H_i}}	(x)&= \frac{2m^m}{\Gamma(m)}x^{2m-1}\exp\left(-mx^2\right),
\end{align}
for $x\geq 0$. 
{\color{black}{In \eqref{eq:f_H}, the Gamma function is defined as $\Gamma(m) = (m-1)!$ for positive integer $m$.}} 
The Nakagami-$m$ distribution is versatile to model various small-scale fading processes. For instance, when $m=1$ and $m=\frac{(K+1)^2}{2K+1}$, it resorts to the Rayleigh and Rician-$K$ distributions, respectively. Further, by tuning its parameter $m$, it is possible to model the signal fading conditions spanning  from severe to moderate, while making  the distribution fit to empirically measured fading data sets. %

    
     \subsection{Performance Metric}
    We assume that the typical user is served by the nearest satellite. Let ${\bf x}_{1}\in \Phi$ be the location of the nearest satellite, the signal-to-interference-puls-noise (SINR) experienced by the typical receiver located at ${\bf u}_1=(0,0,R_{\sf E})$ is:
  \begin{align}
  	{\sf SINR}=\frac{G_1PH_1\|{\bf x}_{1}-{\bf u}_1\|^{-\alpha}}{\sum_{{\bf x}_i \in \Phi\cap\mathcal{A}/\{{\bf x}_{1}\}} G_iPH_i\|{\bf x}_{i}-{\bf u}_1\|^{-\alpha}+\sigma^2},
  \end{align} 
  where $\sigma^2$ denotes the noise power, $H_i$ denotes fading power from satellite $i$ to the receiver.     
     
 First, we characterize the coverage probability when there is at least one satellite exists in $\mathcal{A}$, i.e., $\Phi(\mathcal{A})>0$.  By conditioning on $\Phi(\mathcal{A})>0$, we compute the average of the coverage probabilities over all possible locations of satellites in the spherical cap area $\mathcal{A}$. Such conditional coverage probability is
\begin{align}
 & P^{\sf  cov}_{{\sf SINR}|\Phi(\mathcal{A})>0} (\gamma; \lambda, \alpha,R_{\sf S},m) \nonumber \\
 &=\mathbb{P}\left[{\sf SINR}\geq \gamma \mid~ \Phi(\mathcal{A})>0 \right]  \nonumber\\
&= \! \mathbb{P}\left[\!\frac{H_1\|{\bf x}_{1}-{\bf u}_1\|^{-\alpha}}{\sum_{{\bf x}_i \in \Phi\cap\mathcal{A}/\{{\bf x}_{1}\}} {\bar G}_{i}H_i\|{\bf x}_{i}\!-\!{\bf u}_1\|^{-\alpha}\!+\!{\bar \sigma}^2}\geq \gamma
\middle| \Phi(\mathcal{A})>0 \right]\!,\label{eq:cond_covp}
\end{align}
where ${\bar\sigma}^{2}=\frac{\sigma^2}{PG_1}$ and ${\bar G}_{i}=\frac{G_i}{G_1}<1$ are the normalized noise power and antenna gains. {\color{black}{We note that $\bar G = \bar G_i = \bar G_j$ for $i, j \neq 1$.}} This conditional coverage probability corresponds to the distribution of the SINR experienced by the typical receiver when any satellite BS exists in the spherical cap area $\mathcal{A}$.  This conditional coverage probability also serves as a stepping stone in the derivation of the coverage probability.

 By marginalizing over the binary distribution of $\mathbb{P}[\Phi(\mathcal{A})>0]$ and $\mathbb{P}[\Phi(\mathcal{A})=0]$, we obtain the coverage probability averaged over all possible satellite BSs' geometries, namely,
\begin{align}
 & P^{\sf  cov}_{\sf{SINR}} (\gamma; \lambda, \alpha,R_{\sf S},m) \nonumber \\
 & = P^{\sf  cov}_{{\sf{SINR}}|{\Phi}(\mathcal{A})>0} (\gamma; \lambda, \alpha,R_{\sf S},m) \mathbb{P}[\Phi(\mathcal{A})>0].\label{eq:covp}
\end{align}
This distribution takes the effect of the satellite availability at the typical receiver.

  	 \vspace{0.1cm}
{\bf Remark 1:} Our approach to compute the conditional coverage probability  $P^{\sf  cov}_{{\sf{SINR}}|\Phi(\mathcal{A})>0} (\gamma; \lambda, \alpha,R_{\sf S},m)$ differs from the existing method in \cite{okati:tcom:20}. In \cite{okati:tcom:20}, conditioned on $\Phi(\mathcal{A})=n$, the conditional coverage probability is first computed as
\begin{align}
	P^{\sf  cov}_{{\sf SINR}|\Phi(\mathcal{A})=n} (\gamma; \lambda, \alpha,R_{\sf S},m).
\end{align}  Then, by marginalizing $n$ according to the Poisson distribution in \eqref{eq:poisson}, the coverage probability in \eqref{eq:covp} can also be obtained as
\begin{align}
	 & P^{\sf  cov}_{\sf{SINR}} (\gamma; \lambda, \alpha,R_{\sf S},m) \nonumber \\
	 & =\sum_{n=0}^{\infty}P^{\sf  cov}_{{\sf SINR}|\Phi(\mathcal{A})=n} (\gamma; \lambda, \alpha,R_{\sf S},m) \mathbb{P}[\Phi(\mathcal{A})=n].
\end{align}
This computational approach makes the analysis difficult because the nearest satellite distance distribution is not tractable to integrate the Laplace of the aggregated interference power. Then, the infinite sum also introduces additional computational complexity.  Furthermore, our coverage probability computation differs from \cite{hourani:wcl:21, hourani:wcl:21_2}, in which the satellite spatial distributions are modeled by the PPPs in an infinite two-dimensional and use them with appropriate density matching in the spherical cap, i.e., a finite space. Although this approximation with the density matching can be tight in a dense satellite network, it becomes loose for a low-density satellite network. Our approach, however, is exact irrelevant to satellite densities. 

{\bf{Remark 2:}}
In our analysis, we implicitly assume uniform traffic demands on the earth. 
To capture non-uniform traffic demands, we need to model a user distribution as a non-homogeneous point process that breaks Slivnyak's theorem down. 
Different traffic demands can be captured into the analysis as a form of a load. To be specific, denoting that a user density for a particular region $\mathcal{C}$ as $\lambda_{\sf u}(\mathcal{C})$, the achievable rate per each user in $\mathcal{C}$ is characterized as
\begin{align} \label{eq:rev1:load}
    R = \frac{W_{}}{L(\mathcal{C})} \log \left(1 + {\text{SINR}}_{} \right).
\end{align} 
We note that $W$ is the operating bandwidth and $L(\mathcal{C})$ is the load of the corresponding region, i.e., the number of users associated with a satellite. The distribution of the load is approximately obtained as \cite{singh:twc:13}
\begin{align} \label{eq:rev1:load_dist}
    & \mathbb{P}[L= n] \nonumber \\
    & \simeq \frac{3.5^{3.5}}{n!} \frac{\Gamma(n+4.5)}{\Gamma(3.5)} \left( \frac{\lambda_{\sf u}(\mathcal{C}) }{\lambda} \right)^n \times \left(3.5 + \frac{\lambda_{\sf u}(\mathcal{C}) }{\lambda}  \right)^{-(n + 4.5)},
\end{align}
where $\lambda$ is the satellite density. If there are many users that want to communicate with satellites, i.e., high traffic demands, we have large $\lambda_{\sf u}(\mathcal{C})$. Accordingly, the load $L(\mathcal{C})$ also increases; therefore, the effective bandwidth allocated to each user decreases. This fact results in decreasing the per-user rate.



\section{Coverage Probability Analysis }
This section provides an exact expression for the coverage probability and then derives an upper and lower bound of the coverage probability in analytically tractable forms. 

\subsection{Statistical Properties}

Before providing the general expressions for the coverage probability, we first introduce two vital statistical properties of PPPs on the surface of a sphere, which will be the foundation of our analysis in the sequel.

 \begin{lemma}[The probability of satellite-visibility]\label{lem1}
  The probability that any satellite is visible at the typical receiver is given by
\begin{align}
	\mathbb{P}\left[\Phi(\mathcal{A})>0\right]=1-\exp\left(-\lambda 2\pi(R_{\sf S}-R_{\sf E})R_{\sf S}  \right). 
\end{align}
 \end{lemma}

 \begin{proof}
 Since the locations of satellites follow PPPs on the surface of a sphere, it is sufficient to compute the probability that there is no satellite existing in a typical spherical cap $\mathcal{A}$, namely,
 \begin{align}
 		\mathbb{P}\left[\Phi(\mathcal{A})=0\right]&=\exp\left(-\lambda(|\mathcal{A}|)\right) \nonumber\\
 		&=\exp\left(-\lambda 2\pi(R_{\sf S}-R_{\sf E})R_{\sf S}  \right),\label{eq:void2}
 \end{align}
 where $|\mathcal{A}|= 2\pi(R_{\sf S}-R_{\sf E})R_{\sf S}$ is the Lebesgue measure of set $\mathcal{A}$ as in \eqref{eq:area}. 
 \end{proof}

 \begin{lemma}[The conditional nearest satellite distance distribution]\label{lem2}
 	 Let $R=\min_{{\bf x}_i\in \Phi \cap\mathcal{A}}\|{\bf x}_i-{\bf u}_1\|_2$ be the nearest distance from the typical user's location ${\bf u}_1=(0,0,R_{\sf E})$ to a satellite in $\Phi \cap \mathcal{A}$. Then, the PDF of $R$ is  
\begin{align}
 	f_{R|\Phi(\mathcal{A})>0}(r) =\!\begin{cases}\!
\nu(\lambda,R_{\sf S})re^{-\lambda \pi \frac{R_{\sf S}}{R_{\sf E}}r^2}  &{\rm for}~R_{\rm min}\leq r \leq R_{\rm max},\\
0& {\rm otherwise},
\end{cases} 
 \end{align} 
where $\nu(\lambda,R_{\sf S})=2\pi \lambda \frac{R_{\sf S}}{R_{\sf E}}\frac{ e^{\lambda \pi \frac{R_{\sf S}}{R_{\sf E}} \left(R_{\sf S}^2-R_{\sf E}^2\right) }}{e^{2\lambda \pi R_{\sf S}(R_{\sf S}-R_{\sf E})}-1}$, $R_{\rm min}=R_{\sf S}-R_{\sf E}$, and $R_{\rm max}=\sqrt{R_{\sf S}^2-R_{\sf E}^2}$.
  \end{lemma}
 \proof 
See Appendix \ref{proof:lem2}.
\endproof
The derived conditional distribution of the nearest satellite distance is a key instrument for the coverage probability analysis. It is instructive to compare the derived nearest distance distribution with the time-horned nearest distance distribution for a PPP on $\mathbb{R}^2$ with the same density $\lambda$, i.e., the Rayleigh distribution, 
\begin{align}
	f_{R}(r)= 2\pi \lambda r e^{-\lambda \pi r^2},
\end{align}
for $0\leq r\leq \infty$. The derived nearest distribution can be interpreted as a truncated Rayleigh distribution on the finite support $r\in[R_{\rm min}, R_{\rm max}]$ by conditioning $\Phi(\mathcal{A})>0$. 

{\color{black}{
{\textbf{Remark 3}}: The obtained truncated Rayleigh distribution is reduced to the original Rayleigh distribution by expanding a finite support to an infinite space. Specifically, we assume that $R_{\sf S} = cR_{\sf E}$, wherein $R_{\sf S} \rightarrow \infty, c \rightarrow 1,  R_{\sf S} (1-c) \rightarrow 0, R_{\sf S}^2 (1-c^2) \rightarrow \infty$. Then we have $R_{\rm min} \rightarrow 0$ and $R_{\rm max} \rightarrow \infty$; thereby the support becomes infinite. Under these premises, the truncated Rayleigh distribution is 
\begin{align}
    & f_{R|\Phi(\mathcal{A})>0}(r) \nonumber \\
    &= 2\pi \lambda \frac{R_{\sf S}}{R_{\sf E}}\frac{ e^{\lambda \pi \frac{R_{\sf S}}{R_{\sf E}} \left(R_{\sf S}^2-R_{\sf E}^2\right) }}{e^{2\lambda \pi R_{\sf S}(R_{\sf S}-R_{\sf E})}-1} re^{-\lambda \pi \frac{R_{\sf S}}{R_{\sf E}}r^2} \nonumber \\
    & = 2\pi \lambda c \frac{ e^{\lambda \pi c R_{\sf S}^2\left(1 - c^2\right) }}{e^{\frac{2}{1+c}\lambda \pi R_{\sf S}^2(1-c^2)}-1} re^{-\lambda \pi c r^2} \\
    & \rightarrow 2\pi \lambda r e^{-\lambda \pi r^2},
\end{align}
which shows that the truncated Rayleigh distribution boils down to the original Rayleigh distribution by expanding the support to an infinite space. 
}}


  \begin{lemma} 
 \label{lem3}
{\color{black}{Define the aggregated interference power as
\begin{align}
    I_r = \sum_{{\bf x}_i \in \Phi\cap\mathcal{A}/\{{\bf x}_{1}\}} {\bar G}_{i}H_i\|{\bf x}_{i}\!-\!{\bf u}_1\|^{-\alpha}.
\end{align}
}}
Then, the conditional Laplace transform of the aggregated interference power is 
  \begin{align}
 & \mathcal{L}_{I_{r}|\Phi(\mathcal{A})>0}(s) \nonumber \\
 & = \exp \left(\!- \lambda   \pi  \frac{R_{\sf S}}{R_{\sf E}} \!   \left( \!\frac{{\bar G}_is}{m}\!\right)^{\!\frac{2}{\alpha}} \! \int_{\left(\! \frac{{\bar G}_is}{m}\!\right)^{-\frac{2}{\alpha}}r^2 }^{\left(\! \frac{{\bar G}_is}{m}\!\right)^{-\frac{2}{\alpha}}R_{\rm max}^2  } 1\!-\! \frac{1}{\left(1 \!+\!  u^{\!-\frac{\alpha}{2}} \right)^m}{\rm d} u   \!\right). \label{eq:Laplace}
\end{align}
\end{lemma}
\proof 
See Appendix \ref{proof:lem3}.
\endproof

\subsection{Exact Expression}
The following theorem is our main technical result of this paper.
\begin{theorem}\label{Th1}
In the interference-limited regime, i.e., $I_r \gg {\bar\sigma}^{2}$, the coverage probability of the typical receiver is obtained as 
\begin{align}
&P^{ {\sf cov}}_{{\sf SIR}} (\gamma; \lambda, \alpha,R_{\sf S},m) \nonumber\\
&= \sum_{k=0}^{ m -1}\!\frac{ m^k\gamma^k (-1)^{k}}{k!} \nonumber \\
& \mathbb{E}\!\left[{r}^{\alpha k} \left.\!{\frac{\d^k\mathcal{L}_{{  I}_{r|\Phi(\mathcal{A})>0} }(s)}{\d s^k}} \right|_{s= m\gamma  {r}^{\alpha}}\!\! \!\!\left.\right| \Phi(\mathcal{A})\!>\!0 \right] \left(1- e^{-\lambda 2\pi(R_{\sf S}-R_{\sf E})R_{\sf S}  }\right)  \\
&= \sum_{k=0}^{ m -1}\!\frac{ m^k\gamma^k (-1)^{k}}{k!} \nonumber \\
& \int_{R_{\rm min}}^{R_{\rm max}} {r}^{\alpha k} \left.\!{\frac{\d^k\mathcal{L}_{{  I}_{r|\Phi(\mathcal{A})>0} }(s)}{\d s^k}} \right|_{s= m\gamma  {r}^{\alpha}}\!\! \!\! \nu(\lambda,R_{\sf S})re^{-\lambda \pi \frac{R_{\sf S}}{R_{\sf E}}r^2} {\rm d} r \nonumber \\
& \left(1- e^{-\lambda 2\pi(R_{\sf S}-R_{\sf E})R_{\sf S}  }\right), \label{eq:Th1} 
	\end{align}
where $\nu(\lambda, R_{\sf S})$ is defined in Lemma \ref{lem2} and $\mathcal{L}_{{  I}_{r|\Phi(\mathcal{A})>0} }(s)$ is obtained in Lemma \ref{lem3}.
\end{theorem} 

\begin{proof}
	See Appendix \ref{proof:Th1}.
\end{proof}

In Fig.~\ref{Fig1-1}, we show that the analytical expressions obtained in Theorem \ref{Th1} exactly matches with the numerical coverage probability for all the cases of $m$. Although the coverage probability expression in Theorem \ref{Th1} is exact and holds full generality in all relevant system parameters, it is rather unrefined to understand how the relevant parameters impact the coverage probability.  The difficulty in deriving a more refined expression arises from computing the derivatives of the Laplace transform of the aggregated interference power, which also entails the integral. Then, the additional marginalization is required for the nearest distribution as per Lemma \ref{lem2}. 

{\color{black}{
{\textbf{Remark 4}}: 
We note that Theorem \ref{Th1} assumes the interference-limited regime, wherein the aggregate interference power is a dominant factor to determine the coverage probability. If we assume the noise-limited regime, i.e., ${\bar{\sigma}}^2 \gg I_r$, the coverage probability of the typical user is easily obtained as 
    \begin{align}
 &P^{ {\sf cov}}_{{\sf SNR}|\Phi(\mathcal{A})>0} (\gamma; \lambda, \alpha,R_{\sf S},m) \nonumber\\
&\stackrel{}{=}\nu(\lambda, R_{\sf S})\!\!\int_{R_{\rm min}}^{R_{\rm max}}  \sum_{k=0}^{ m -1}\frac{ m^k\gamma^k {r}^{\alpha k}}{k!} \bar \sigma^{2k} e^{-m r^{\alpha} \bar \sigma^2} \times     re^{-\lambda \pi \frac{R_{\sf S}}{R_{\sf E}}r^2} {\rm d} r.
\end{align}
In Fig.~\ref{Fig1-1}, we show that if satellites are densely deployed as $\lambda|\mathcal{A}|=10$, the SIR coverage and the SINR coverage are almost equal, meaning that the noise effect is negligible in the networks. This justifies the interference-limited assumption in Theorem \ref{Th1}. 
The used noise parameters are as follows. The noise spectral density: $-174 {\rm dBm}/{\rm Hz}$, the bandwidth: $10 {\rm MHz}$, the transmitted power of satellites: $10 {\rm W}$, the transmit antenna gain: $30 {\rm dBi}$. We also note that all the parameters are referred from \cite{okati:tcom:20}. 
}}

\begin{figure} 
    \centering 
    \includegraphics[width=1\columnwidth]{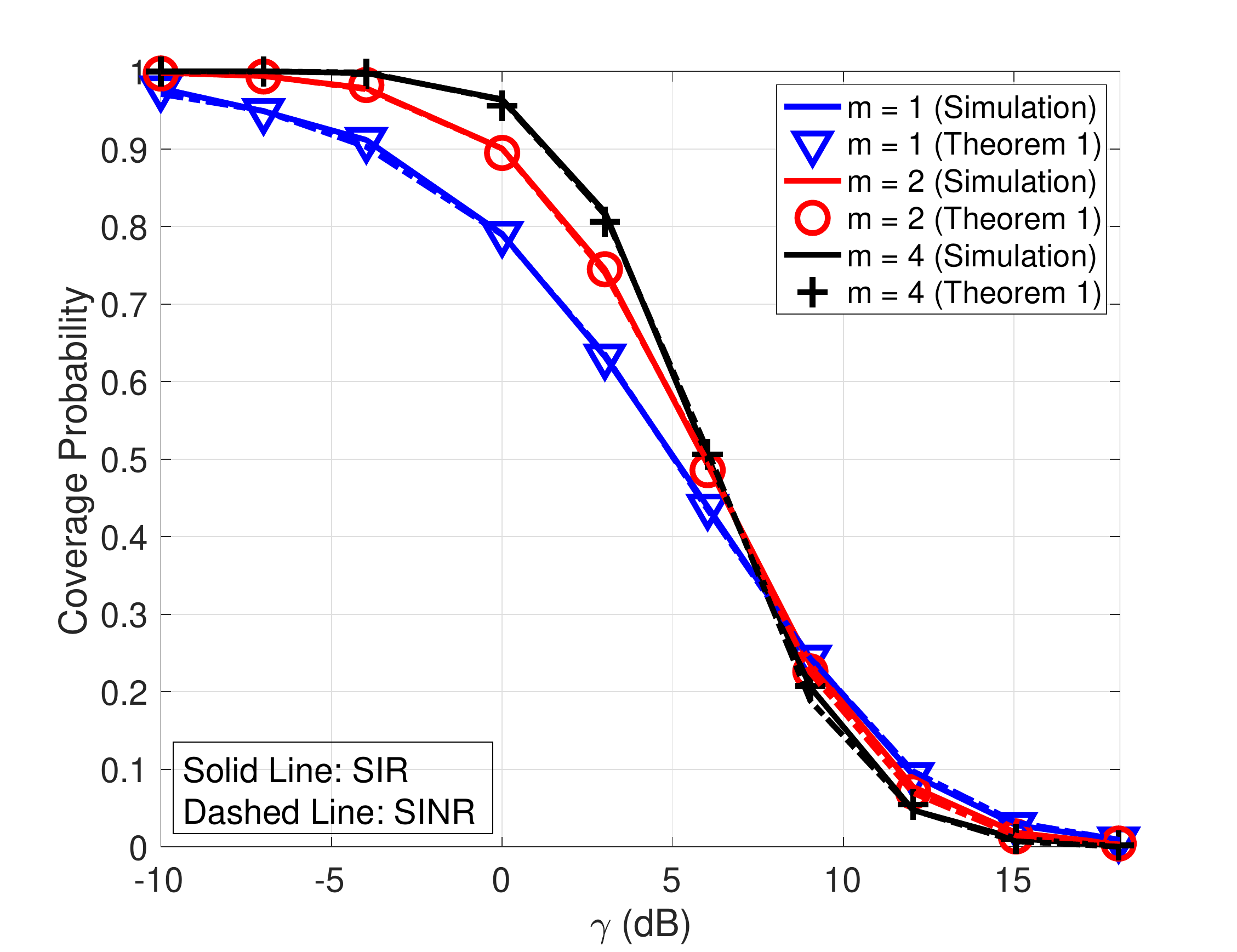}
    \caption{The coverage probability according to the Nakagami fading parameter $m\in \{1,2,4\}$ for fixed $\alpha=2$, ${\bar G}_i=0.1$, $\lambda|\mathcal{A}|=10$, $R_{\sf S}-R_{\sf E}=500$ km, and $m\in \{1,2,4\}$.  }\label{Fig1-1}
\end{figure}
  
\subsection{Upper and Lower Bounds}
To make a more compact characterization than the one in Theorem \ref{Th1}, we provide an upper and lower bound of the coverage probability expression without loss of generality.
 
  \begin{theorem}\label{Th2} 
In the interference-limited regime, the coverage probability is upper and lower bounded as 
  \begin{align}
  P^{{\sf{cov}}, {\rm{B}}}_{{\sf SIR}} (\gamma ;\lambda, \alpha,R_{\sf S},m,\kappa^{\rm L}) &	\leq P^{\sf cov}_{{\sf SIR}} (\gamma ;\lambda, \alpha,R_{\sf S},m) \nonumber\\
 P^{\sf cov}_{{\sf SIR}} (\gamma ;\lambda, \alpha,R_{\sf S},m)   &\leq  P^{{\sf{ cov}}, {\rm{B}} }_{{\sf SIR}} (\gamma ;\lambda, \alpha,R_{\sf S},m, \kappa^{\rm U}),
  \end{align}
where
\begin{align}
  &	P^{{\sf{ cov}}, {\rm{B}}}_{{\sf SIR}} (\gamma ;\lambda, \alpha,R_{\sf S},m,\kappa) \nonumber\\
 &= 2\pi \lambda \frac{R_{\sf S}}{R_{\sf E}}e^{\lambda \pi \frac{R_{\sf S}}{R_{\sf E}} \left(R_{\sf S}-R_{\sf E}\right)^2 }  \sum_{\ell=1}^{m} \binom{m}{\ell}(-1)^{\ell+1 } \nonumber \\
 & \int_{R_{\rm min}}^{R_{\rm max}} \! re^{-\lambda  \frac{R_{\sf S}}{R_{\sf E}} \pi\left[1+\eta(\ell m \kappa \gamma , r; \alpha, R_{\sf S}, m) \right]r^2}\! {\rm d} r, \label{eq:Th2}
 \end{align} 	
  with $\kappa\in [\kappa^{\sf U},\kappa^{\sf L}]$ for $ \kappa^{\sf L}=1$, $ \kappa^{\sf U}=m!^{-\frac{1}{m}}$. Also,  
\begin{align}
 \eta(x, r ;\alpha, R_{\sf S}, m) =\left(\!\frac{{\bar G_i}x}{m}\!\right)^{\!\frac{2}{\alpha}}\int_{\left(\!\frac{{\bar G}_i{ x}}{m}\!\right)^{\!\!-\frac{2}{\alpha}} }^{\left(\!\frac{{\bar G}_i{ x}}{m}\!\right)^{\!\!-\frac{2}{\alpha}}\!\left(\!\frac{R_{\rm max}}{r}\!\right)^2  }\!\! 1\!-\! \frac{1}{\left(1 \!+\!  u^{-\frac{\alpha}{2}} \right)^m}{\rm d} u.\label{eq:eta}
 \end{align}

    \end{theorem}
  \proof 
See Appendix \ref{proof:Th2}.
\endproof
  
   The coverage probability bounds in Theorem  \ref{Th2} are relatively tractable than the former in Theorem \ref{Th1}.  In particular, when $m=1$, i.e., the upper and lower bounds coincide because $\kappa^{\sf U}=1$ for $m=1$. By carefully tuning parameter $\kappa$ between $1$ and $m!^{-\frac{1}{m}}$ depending on the network parameters, we can obtain very tight approximation for the coverage probability as  \begin{align}
  &	P^{{\sf{ cov}}, {\rm{A}}}_{{\sf SIR}} (\gamma ;\lambda, \alpha,R_{\sf S},m,\kappa ) \nonumber\\
 &= 2\pi \lambda \frac{R_{\sf S}}{R_{\sf E}}e^{\lambda \pi \frac{R_{\sf S}}{R_{\sf E}} \left(R_{\sf S}-R_{\sf E}\right)^2 }  \sum_{\ell=1}^{m} \binom{m}{\ell}(-1)^{\ell+1 } \nonumber \\ 
 & \int_{R_{\rm min}}^{R_{\rm max}} \! re^{-\lambda  \frac{R_{\sf S}}{R_{\sf E}} \pi\left[1+\eta(\ell m \kappa  \gamma , r; \alpha, R_{\sf S}, m) \right]r^2}\! {\rm d} r, \label{eq:approx}
 \end{align} 	
  where $1\leq \kappa \leq (m!)^{-\frac{1}{m}}$. Also, the approximation coincides with the exact one in Theorem \ref{Th1}, which is stated in the following corollary.  
   \begin{corollary}\label{Cor1}
 Under the Rayleigh fading propagation, i.e., $m=1$, the approximation restores the exact coverage probability, i.e., 
 \begin{align}
 	 & P^{\sf cov}_{{\sf SIR}} (\gamma; \lambda, \alpha,R_{\sf S},1) \nonumber \\
 	 & = 2\pi \lambda \frac{R_{\sf S}}{R_{\sf E}}e^{\lambda \pi \frac{R_{\sf S}}{R_{\sf E}} \left(R_{\sf S}-R_{\sf E}\right)^2 }  \!\!\!\int_{R_{\rm min}}^{R_{\rm max}} \! \!re^{-\lambda  \frac{R_{\sf S}}{R_{\sf E}} \pi\left[1+\eta(\gamma , r  ; \lambda, \alpha, R_{\sf S}, 1 ) \right]r^2}\! {\rm d} r. \label{avgcovp1}
 	\end{align}
 	
\end{corollary}
 \begin{proof}
 	The proof is direct by setting $m=1$.
\end{proof}

\begin{figure} 
    \centering 
    \includegraphics[width=1\columnwidth]{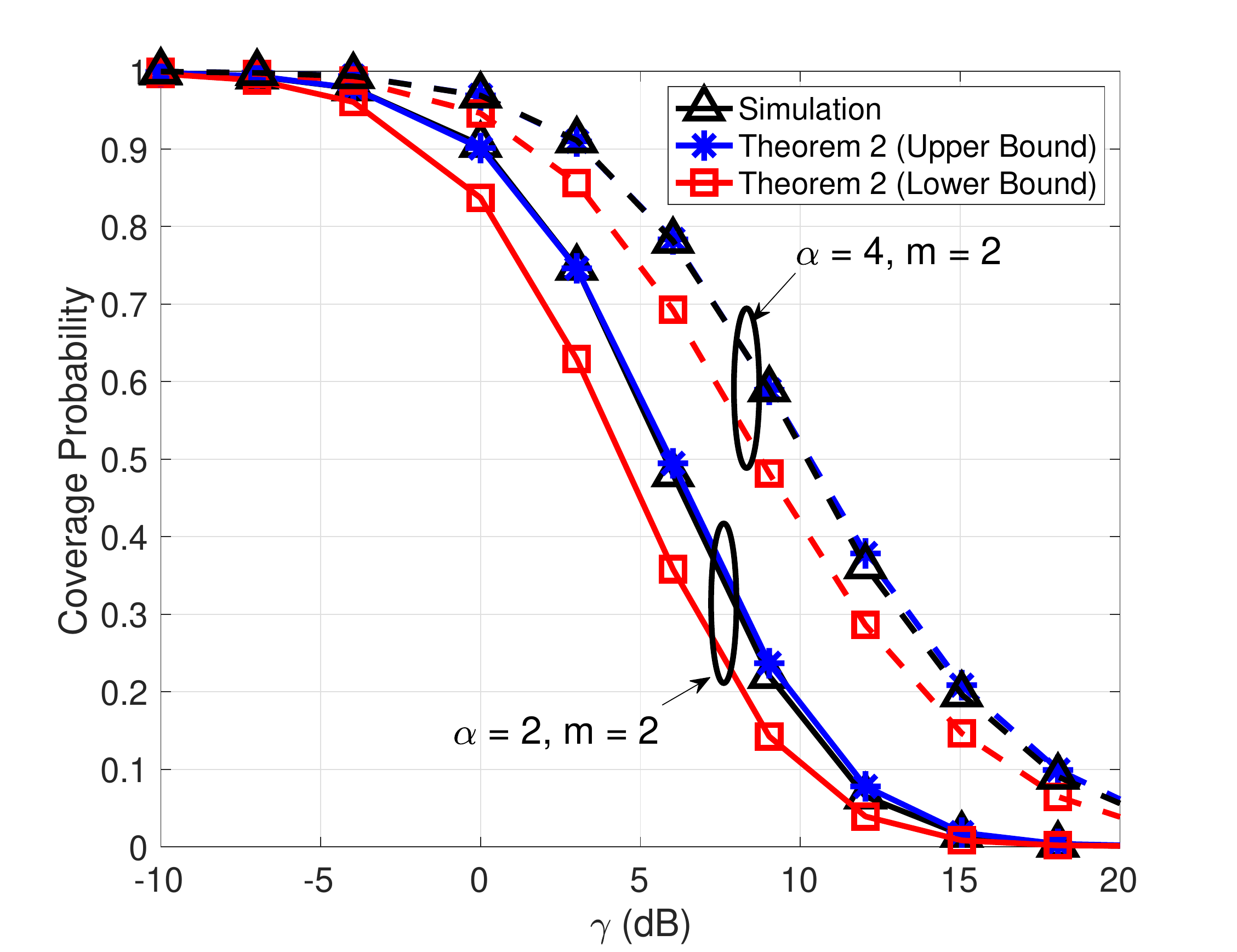}
    \caption{The coverage probability description for verifying Theorem \ref{Th2}. The simulation setups are: $\alpha \in \{2, 4\}$, ${\bar G}_i=0.1$, $\lambda|\mathcal{A}|=10$, $R_{\sf S}-R_{\sf E}=500$ km, and $m = 2$.  }\label{Fig_theorem2}
\end{figure}

As illustrated in Fig. \ref{Fig1-1}, the derived coverage probability expression in \eqref{eq:approx} tightly matches with the exact coverage probability obtained via numerical simulations over the entire range of $\gamma$ of interest and for the values of all relevant system parameters, $\lambda$, $\alpha$, $m$, and $R_{\sf S}$.

  
    \subsection{Lower Bound in {\color{black}{Tractable}} Form}

To make the analysis more tractable, we characterize an approximation to the coverage probability in a tractable form. 
{\color{black}{A main tractability bottleneck in Theorem \ref{Th2} is that the integral term is a function of $r$. 
}}
To resolve this, our approach is to represent the integral term in the conditional Laplace in Lemma \ref{lem3} to be independent of $r$, while making this tight. The following corollary represents a lower bound of the coverage probability in Theorem \ref{Th2} in tractable form.

      \begin{corollary}\label{Cor2} 
  	A lower bound of the coverage probability in {\color{black}{a tractable}} form is given by
{\color{black}{
 \begin{align}
 &	P^{{\sf{ cov}}, {\rm{L}}}_{{\sf SIR}} (\gamma; \lambda, \alpha,R_{\sf S},m) = e^{\lambda  \pi  \frac{R_{\sf S}}{R_{\sf E}} (R_{\sf S} - R_{\sf E})^2 } 
 \sum_{\ell=1}^{m} \! \binom{m}{\ell}(-1)^{\ell+1 } \nonumber \\
 & \left[\frac{e^{-\lambda \pi \frac{R_{\sf S}}{R_{\sf E}} (1+\eta^{\sf U}(\ell m \gamma  ; \alpha, R_{\sf S}, m)) R_{\rm min}^2 }  }{1+\eta^{\sf U}(\ell m \gamma  ; \alpha, R_{\sf S}, m) } -  \right. \nonumber \\
 &\left.  \frac{ e^{-\lambda \pi \frac{R_{\sf S}}{R_{\sf E}}  ( 1 +\eta^{\sf U}(\ell m \gamma  ; \alpha, R_{\sf S}, m) ) R_{\rm max}^2 } }{1+\eta^{\sf U}(\ell m \gamma  ; \alpha, R_{\sf S}, m) } \right],\label{eq:cor2}
 \end{align} 
 }} 
where 
\begin{align}
 \eta^{\sf U}(x; \alpha, R_{\sf S}, m)  =\left(\!\frac{{\bar G_i}x}{m}\!\right)^{\!\frac{2}{\alpha}}\int_{\left(\!\frac{{\bar G}_i{ x}}{m}\!\right)^{\!\!-\frac{2}{\alpha}} }^{\left(\!\frac{{\bar G}_i{ x}}{m}\!\right)^{\!\!-\frac{2}{\alpha}}\!\left(\!\frac{R_{\rm max}}{R_{\rm min}}\!\right)^2  }\!\! 1\!-\! \frac{1}{\left(1 \!+\!  u^{-\frac{\alpha}{2}} \right)^m}{\rm d} u.\label{eq:eta_upp}
 \end{align}
 \end{corollary}
   \proof 
See Appendix \ref{proof:cor2}.
\endproof

From the derived lower bound of the coverage probability, we can see that the coverage probability changes over the density of satellites $\lambda$ for fixed path-loss exponent $\alpha$, satellites' altitude $R_{\sf S}$, and target threshold $\gamma$. This result is in contrast to the coverage probability for terrestrial cellular networks modeled by homogenous PPPs, in which the coverage probability has shown to be invariant with the density of base stations in the interference-limited regime, assuming the standard single-slope path-loss law \cite{andrews:tcom:11,lee:twc:15}.  

It is remarkable that the coverage probability can either increase or decrease depending on the satellite density $\lambda$ as in \eqref{eq:cor2}. As the density becomes smaller,  the conditional coverage probability $P^{{\sf{ cov}}, {\rm L}}_{{\sf SIR}|\Phi(\mathcal{A})>0} (\gamma; \lambda, \alpha,R_{\sf S},m)$ improves, while the probability of satellite-visibility $\mathbb{P}[\Phi(\mathcal{A})>0]=\left(1- e^{-\lambda 2\pi(R_{\sf S}-R_{\sf E})R_{\sf S}  }\right)$ becomes deteriorated. As a result, optimizing the satellite density is the key to maximize the coverage performance in the satellite network design. In the next section, we will characterize the optimal satellite density for given other network design parameters. 
{\color{black}{For convenience, we present the following corollary as a special case of Corollary \ref{Cor2}. }}



    \begin{corollary}\label{Cor3} 
    Under the Rayleigh fading propagation, i.e., $m=1$, the  lower bound of the coverage probability boils down to
    {\color{black}{
 \begin{align}
 &	P^{{\sf{ cov}}, {\rm{L}}}_{{\sf SIR}} (\gamma; \lambda, \alpha,R_{\sf S},1)\nonumber\\
 =& e^{\lambda  \pi  \frac{R_{\sf S}}{R_{\sf E}} (R_{\sf S} - R_{\sf E})^2 } 
 \left[\frac{e^{-\lambda \pi \frac{R_{\sf S}}{R_{\sf E}} (1+\eta^{\sf U}( \gamma  ; \alpha, R_{\sf S}, 1)) R_{\rm min}^2 }  }{1+\eta^{\sf U}( \gamma  ; \alpha, R_{\sf S}, 1) } - \right. \nonumber \\
 & \left.\frac{ e^{-\lambda \pi \frac{R_{\sf S}}{R_{\sf E}}  ( 1 +\eta^{\sf U}( \gamma  ; \alpha, R_{\sf S}, 1) ) R_{\rm max}^2 } }{1+\eta^{\sf U}( \gamma  ; \alpha, R_{\sf S}, 1) } \right].\label{eq:closedform_Rayleigh}
 \end{align} 
We note that when $m = 1$, 
\begin{align}
    & \eta^{\sf U}(x; \alpha, R_{\sf S}, 1) = \left( \frac{{\bar G_i}x}{1} \right)^{ \frac{2}{\alpha}}\int_{\left( \frac{{\bar G}_i{ x}}{1} \right)^{  -\frac{2}{\alpha}} }^{\infty}   1 -  \frac{1}{\left(1  +   u^{-\frac{\alpha}{2}} \right)}{\rm d} u  - \nonumber\\
    & \left( \frac{{\bar G_i}x}{1} \right)^{ \frac{2}{\alpha}}\int_{\left( \frac{{\bar G}_i{ x}}{1} \right)^{  -\frac{2}{\alpha}} \left( \frac{R_{\rm max}}{R_{\rm min}} \right)^2  }^{\infty}   1 -  \frac{1}{\left(1  +   u^{-\frac{\alpha}{2}} \right)}{\rm d} u \\
    &=  \frac{\left({\bar G_i}x\right)^{ \frac{2}{\alpha}}}{\alpha-2}\left[ 2\left( \bar G_i x\right)^{\frac{2}{\alpha}-1} {}_2 F_1 \left(1, 1-\frac{1}{\alpha}, 2-\frac{2}{\alpha}, -\bar G_i x  \right) - \right. \nonumber \\
    &  \left. 2\left(\left({\bar G}_i{ x}\right)^{  -\frac{2}{\alpha}} \left( \frac{R_{\rm max}}{R_{\rm min}} \right)^2  \right)^{1-\frac{\alpha}{2}} \cdot \right. \nonumber \\
    & \left. {}_2 F_1 \left(1, 1-\frac{1}{\alpha}, 2-\frac{2}{\alpha}, -\left( \left({\bar G}_i{ x}\right)^{  -\frac{2}{\alpha}} \left( \frac{R_{\rm max}}{R_{\rm min}} \right)^2\right)^{-\frac{\alpha}{2}} \right) \right],
\end{align}
where $_2F_1(\cdot, \cdot, \cdot, \cdot)$ is a hypergeometric function. 
 }} 
  
    \end{corollary}
     \begin{proof}
 	The proof is direct from \eqref{eq:cor2} by setting $m=1$.
 \end{proof}
 
 Fig. \ref{Fig1-2} shows the tightness of the derived lower bound for the coverage probability in Corollary \ref{Cor2}. As shown in Fig.~\ref{Fig1-2}, the lower bound becomes tight as the density of satellites increases. When the density is high, the difference between $\eta(\gamma; \alpha, R_{\sf S}, 1 )$ and  $\eta^{\sf U}(\gamma; \alpha, R_{\sf S}, 1 )$ becomes less significant to the coverage probability.


\begin{figure} 
    \centering 
    \includegraphics[width=1\columnwidth]{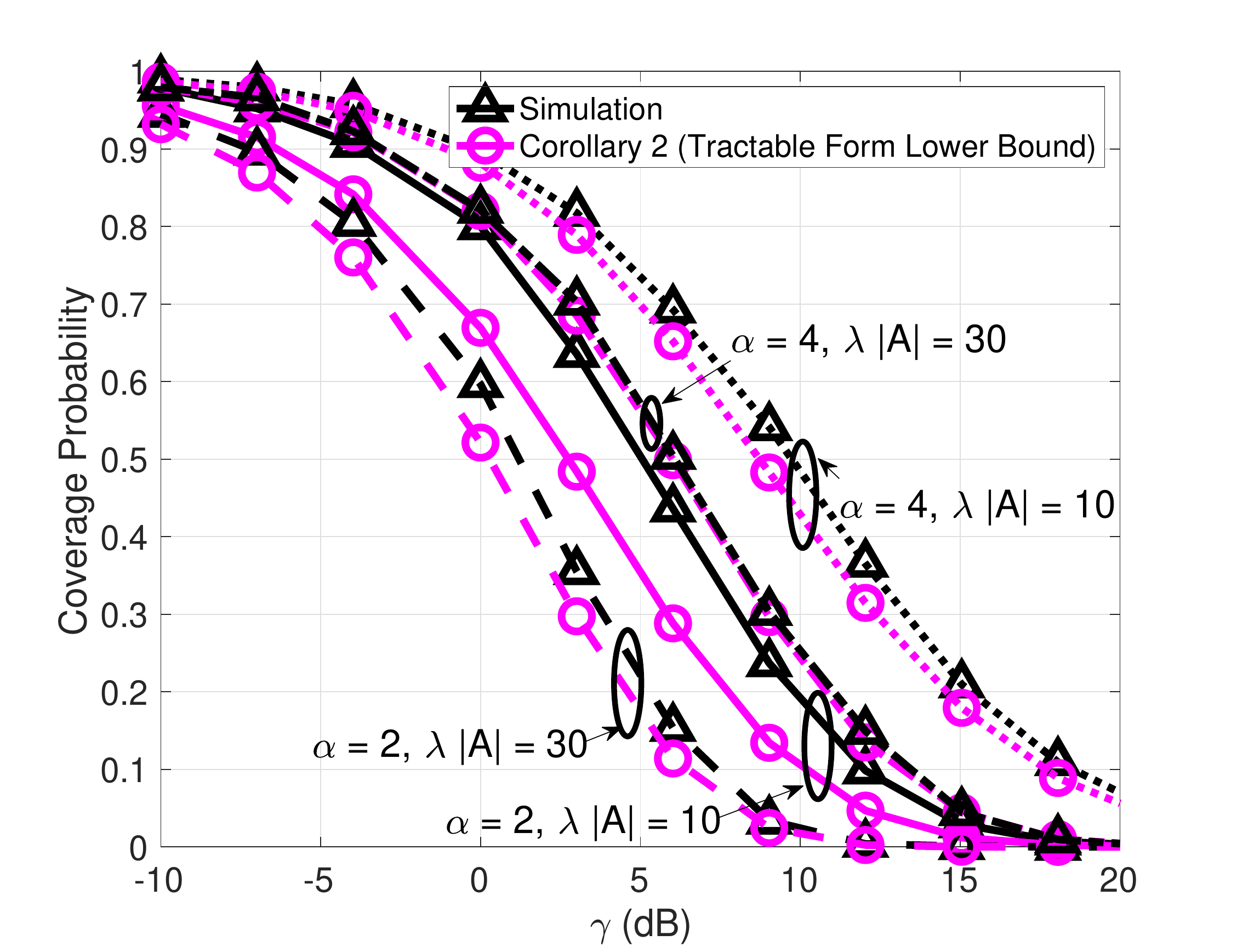}
    \caption{The coverage probability for verifying Corollary \ref{Cor2}. The simulation setups are: $\alpha=\{2,4\}$, ${\bar G}_i=0.1$, $\lambda|\mathcal{A}|\in \{10,30\}$ and $m=1$. }\label{Fig1-2}
\end{figure}

 \subsection{Special Cases and Interpretation}
 

 To shed further light on the significance of the coverage probability expressions in Corollary \ref{Cor2}, we consider certain values of the path-loss exponent $\alpha\in\{2,4\}$ and the Nakagami fading parameter $m\in\{1,2,4\}$. For these cases, we provide more tractable expressions for the coverage probabilities by specifying $\eta^{\sf U}(x; \alpha, R_{\sf S}, m) $ in \eqref{eq:eta_upp} into closed-forms.

{\bf LOS propagation:} The primary case of interest is when the satellites in $\mathcal{A}$ are visible under the premise that the typical receiver's location is outdoor (e.g., the rooftop satellite dish antenna). In this case, the desired and interference signals experience LOS path-loss, i.e., $\alpha=2$. We note that the shape parameter of the Nakagami-$m$ is connected to the Rician $K$-factor as 
\begin{align}
	m=\frac{(K+1)^2}{2K+1},
\end{align}
where $K$ represents the ratio between the powers of the direct and scattered paths in dB scale. For instance, $K=0$ (in dB) corresponds to the Rayleigh fading case.
When $m=2$, i.e., the LOS fading component is approximately $K=1+\sqrt{2}\simeq 2.414$ dB larger than the scattered ones, the integral term defined in \eqref{eq:eta_upp} reduces to 
  	    \begin{align}
 \eta^{\sf U}(x/\bar G; 2, R_{\sf S}, 2 ) =\frac{x^2}{x+\frac{R_{\sf S}+R_{\sf E}}{R_{\sf S}-R_{\sf E}}
}-\frac{x^2}{x+1}+2 x\log \left(\frac{x+\frac{R_{\sf S}+R_{\sf E}}{R_{\sf S}-R_{\sf E}}
}{x+1}\right).   \label{eq:22}
     	\end{align} 
As the LOS component in fadings becomes predominant, i.e., $m=4$, which corresponds to $K=3+2\sqrt{3}\simeq 6.464$ dB, we obtain
  	    \begin{align}
& \eta^{\sf U}(x/\bar G; 2, R_{\sf S}, 4) \nonumber \\
&=  \frac{x^2 \left(13x^2+30 x \frac{R_{\sf S}+R_{\sf E}}{R_{\sf S}-R_{\sf E}}+18 \left(\frac{R_{\sf S}+R_{\sf E}}{R_{\sf S}-R_{\sf E}}\right)^2\right)}{3\left(x+\frac{R_{\sf S}+R_{\sf E}}{R_{\sf S}-R_{\sf E}}\right)^3}\nonumber\\
&~~~-\frac{x^2 \left(13 x^2+30x+18\right)}{3(x+1)^3}  + 4 x \log \left(\frac{x+\frac{R_{\sf S}+R_{\sf E}}{R_{\sf S}-R_{\sf E}}}{x+1}\right)\!.  \label{eq:24}
     	\end{align}  
It is worth mentioning that since $\eta^{\sf U}(x; 2, R_{\sf S}, 4 )>\eta^{\sf U}(x; 2, R_{\sf S}, 2)$ for $\gamma \geq 0$, the integral term $\eta^{\sf U}(x; 2, R_{\sf S}, m)$ decreases as the LOS component in fadings becomes stronger.  Combined this fact with the lower bound expression in Corollary \ref{Cor2}, we can deduce that the LOS propagation is beneficial to improve the coverage performance, which agrees with our intuition.


{\bf  NLOS propagation:} The NLOS propagation model with path-loss exponent $\alpha=4$ and the rich-scattering environment $m=1$ is another case of interest, assuming the typical receiver is placed indoors. In this case, the exponent of the interference Laplace reduces to
   \begin{align}
\eta^{\sf U}(x/\bar G ; 4, R_{\sf S}, 1)
       &= \sqrt{x} ~{\rm arctan}\left( \frac{ \frac{2R_{\sf E}}{R_{\sf S}-R_{\sf E}}  \sqrt{x} }{ \frac{R_{\sf S}+R_{\sf E}}{R_{\sf S}-R_{\sf E}}+ x }\!\right).\label{eq:41}
  	\end{align}
  	To see the effect of the path-loss exponent, we also consider the case of $\alpha=2$ and $m=1$. The integral in \eqref{eq:eta_upp} becomes
   \begin{align}
\eta^{\sf U}(x/\bar G; 2, R_{\sf S}, 1)
  =x \ln\left( \frac{\frac{R_{\sf S}+R_{\sf E}}{R_{\sf S}-R_{\sf E}}+x}{1+x}\right).   \label{eq:21}
  	\end{align}
By comparing $\eta^{\sf U}(x; 4, R_{\sf S}, 1)$ in \eqref{eq:41} and $\eta^{\sf U}(x; 2, R_{\sf S}, 1)$ in \eqref{eq:21}, we observe that the coverage probability becomes worse as the path-loss exponent $\alpha$ decreases because $\eta^{\sf U}(x; 2, R_{\sf S}, 1)$ is larger than $\eta^{\sf U}(x; 4, R_{\sf S}, 1)$ for wide range of $x$ for fixed $R_{\sf S}$. Moreover, from Corollary \ref{Cor2},  the coverage performance becomes worse when increasing $\eta(x ; 2, R_{\sf S}, 1)$. This fact implies that the small path-loss exponent, such as the LOS propagation environment, is detrimental to the coverage performance. This observation is interesting because the LOS propagation enhances the coverage probability by reducing the small-scale fading effects. The NLOS propagation is preferred for the coverage enhancement because the high path-loss exponent decreases more co-channel interference power than the LOS case. Notwithstanding, this result can be the opposite if we adopt a more sophisticated path-loss model featuring LOS and NLOS propagations for the desired and interference signals separately with high probability.

\begin{figure} 
    \centering 
    \includegraphics[width=1\columnwidth]{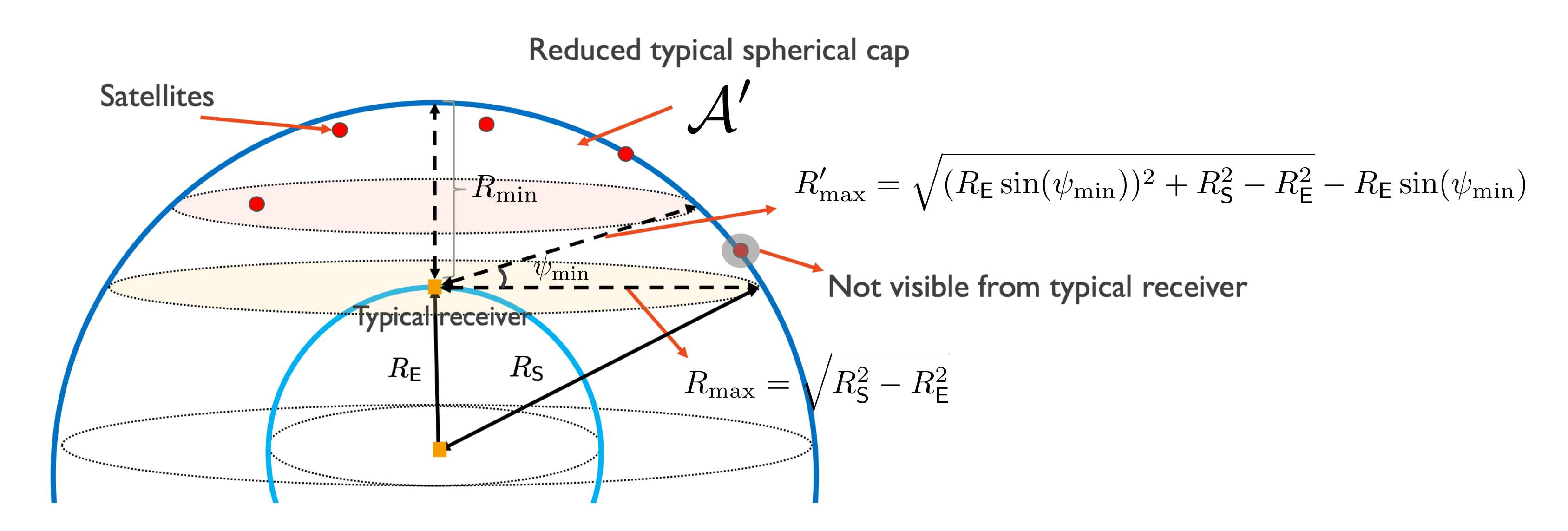}
    \caption{Illustration of reduced spherical cap caused by limited visibility with $\psi_{\sf min}$. 
    }\label{fig:geo_reduced}
\end{figure}

{\color{black}{
{\bf Limited visibility:} In practice, the visibility of the typical receiver can be limited. Specifically, a satellite whose elevation angle is smaller than $\psi_{\rm min}$ may not be visible to the typical receiver. This results in the reduced typical spherical cap $\mathcal{A}'$ and the reduced maximum distance $R_{\rm max}'$. Please see Fig.~\ref{fig:geo_reduced} for a detailed illustration. As seen in Fig.~\ref{fig:geo_reduced}, if $\psi_{\rm min} = 0$, the limited visibility case goes back to the full visibility case. To reflect this limited visibility case into our analysis, we first compute the reduced maximum distance $R_{\rm max}'$ and the area of the reduced typical spherical cap $\mathcal{A}'$ and as a function of $\psi_{\rm min}$ as follows:
\begin{align}
    & R_{\rm max}' = \sqrt{(R_{\sf E} \sin (\psi_{\rm min}))^2 + R_{\sf S}^2 - R_{\sf E}^2 } - R_{\sf E} \sin (\psi_{\rm min}), \\
    & |\mathcal{A}'| = 2\pi R_{\sf S}\left(R_{\sf S} - R_{\sf E} - R_{\rm max}' \cos(\pi/2 - \psi_{\rm min})\right).
\end{align}
By replacing $R_{\rm max}$ by $R_{\rm max}'$ and $|\mathcal{A}|$ by $|\mathcal{A}'|$, our analytical results are readily extended for the limited visibility case. 
}}

\section{Optimal Satellite Density}


In this section, we derive an optimal satellite density that maximizes the lower bound of the coverage probability in Corollary \ref{Cor3}.  The following theorem illuminates how the optimal satellite density is determined as a function of satellite height $R_{\rm min}$ and path-loss exponent $\alpha$.

   \begin{theorem}\label{Th3}
   	For given $\alpha$, $R_{\sf S}=R_{\rm min}+R_{\sf E}$, and $\gamma$, we define the optimal density that maximizes the lower bound of the coverage probability as 
   	\begin{align}
   		\lambda^{\star} =\argmax_{\lambda\geq 0}P^{{\sf{ cov}}, {\rm{L}}}_{{\sf SIR}} (\gamma ;\lambda, \alpha,R_{\sf S},1).
   	\end{align}
   Then, such optimal density is
   	\begin{align}
   	\lambda^{\star} = \frac{\ln\left(1 + \frac{2 R_{\sf E}\left[1+\eta^{\sf U}(\gamma; \alpha, R_{\sf S}, 1 )\right]}{ \eta^{\sf U}(\gamma; \alpha, R_{\sf S}, 1 )R_{\rm min}}\right)}{\left[1+\eta^{\sf U}(\gamma; \alpha, R_{\sf S}, 1 )\right] (2\pi R_{\sf S} R_{\rm min})}. \label{eq:opt_density}
   	\end{align}

   \end{theorem}

  \begin{proof} 
 
See Appendix \ref{proof:Th3}.
  \end{proof}

%
 
Theorem \ref{Th3} suggests that the optimal density $\lambda^{\star}$ diminishes as the altitude of the satellite $R_{\rm min}$ increases.  This result implies that it is necessary to deploy more satellites for LEO networks to enhance the coverage performance, which agrees with the conventional wisdom obtained from extensive simulation studies \cite{vatalaro:jsac:95, mokhtar:wcl:20,seyedi:commlett:12}.

To provide a better understanding for the result in Theorem \ref{Th3},  it is more informative to characterize the optimal average number of satellites  by multiplying $\lambda^{\star}$ to the area of the spherical cap $|\mathcal{A}|=2 \pi R_{\sf S} R_{\rm min}$, namely, 
   	\begin{align}
   		 \lambda^{\star}|\mathcal{A}|  =  \frac{ \ln\left(1 + \frac{2 R_{\sf E}\left[1+\eta^{\sf U}(\gamma; \alpha, R_{\sf S}, 1 )\right]}{ \eta^{\sf U}(\gamma; \alpha, R_{\sf S}, 1 )R_{\rm min}}\right)}{\left[1+\eta^{\sf U}(\gamma; \alpha, R_{\sf S}, 1 )\right] }. \label{eq:trade-off}  
   	\end{align}
This fact elucidates the optimal trade-off between the average number of satellites and altitude $R_{\rm min}$. From this trade-off, we can provide valuable guidance in the network design. For instance, the average number of satellites in the typical spherical cap decreases with the altitude logarithmically to maximize the coverage performance. In particular, for VLEO satellite networks, in which $ R_{\rm min} \ll R_{\sf E}(=6350~ {\rm km})$, the optimal number of satellites approximately scales as 
   	\begin{align}
		 \lambda^{\star}|\mathcal{A}|  \propto  \frac{1}{[1 + \eta^{\sf U}(\gamma; \alpha, R_{\sf S}, 1 )]}\ln \left( \frac{ 1}{ R_{\rm min}} \right).  
   	\end{align}
In the shallow altitude regime, that average number of satellites can scale down logarithmically with the altitude. Nonetheless, deploying more satellites requires high costs; a satellite network operator needs to optimize the trade-off between the deploying cost and the coverage performance.

\section{Simulation Results}
This section provides numerical results to validate the derived coverage probabilities and assess the impact of important system parameters, including the density of satellites, path-loss exponent, and altitude. In particular, we also illustrate how many satellites are needed in an average sense to maximize the coverage probability according to the satellite altitude. In our simulation, we set ${ \bar G}_i=0.1$, i.e., the desired link experiences 10 dB higher antenna gain compared to those from interfering links thanks to the transceiver's beam alignment procedure.

{\color{black}{
Note that we use MATLAB simulation for obtaining the coverage probability. 
We first build a 3D satellite network geometry as in Fig.~\ref{fig:geo}, wherein the typical user is located at $(0,0,R_{\sf E})$. Thereafter, we identify the typical spherical cap $\mathcal{A}$ and distribute points in $\mathcal{A}$ according to a PPP with density $\lambda |\mathcal{A}| $. Then we make a channel fading on each link. 
}}




{\color{black}{
{\bf Comparison of PPP and BPP:} To understand the impact on the spatial distributions of satellite networks, we compare the coverage probability of the proposed PPP model with that of the existing BPP model used in \cite{okati:tcom:20,talgat:commlett:20,talgat:commlett:21} in Fig.~\ref{Fig_ppp_bpp}. We note that $\lambda|\mathcal{A}|$ is the average number of satellites in $\mathcal{A}$ in a PPP model, while $N_{\text {BPP}}$ is the exact number of satellites in $\mathcal{A}$ in a BPP model. As seen in Fig. \ref{Fig_ppp_bpp}, the coverage probability performances of the two models are tightly matched when $\lambda|\mathcal{A}|=N_{\text{BPP}} = 10$. On the contrary to that, the two models provide inconsistent coverage probability performances when the satellite density is low, i.e., $\lambda|\mathcal{A}|= N_{\text{BPP}}=  2$. This result is aligned with that the BPP only can tightly approximate the PPP for a dense network. Nonetheless, the BPP model is limited in capturing the interplay between the satellite access probability and the conditional coverage probability in the low-density regime. 
}}

\begin{figure} 
    \centering 
    \includegraphics[width=1\columnwidth]{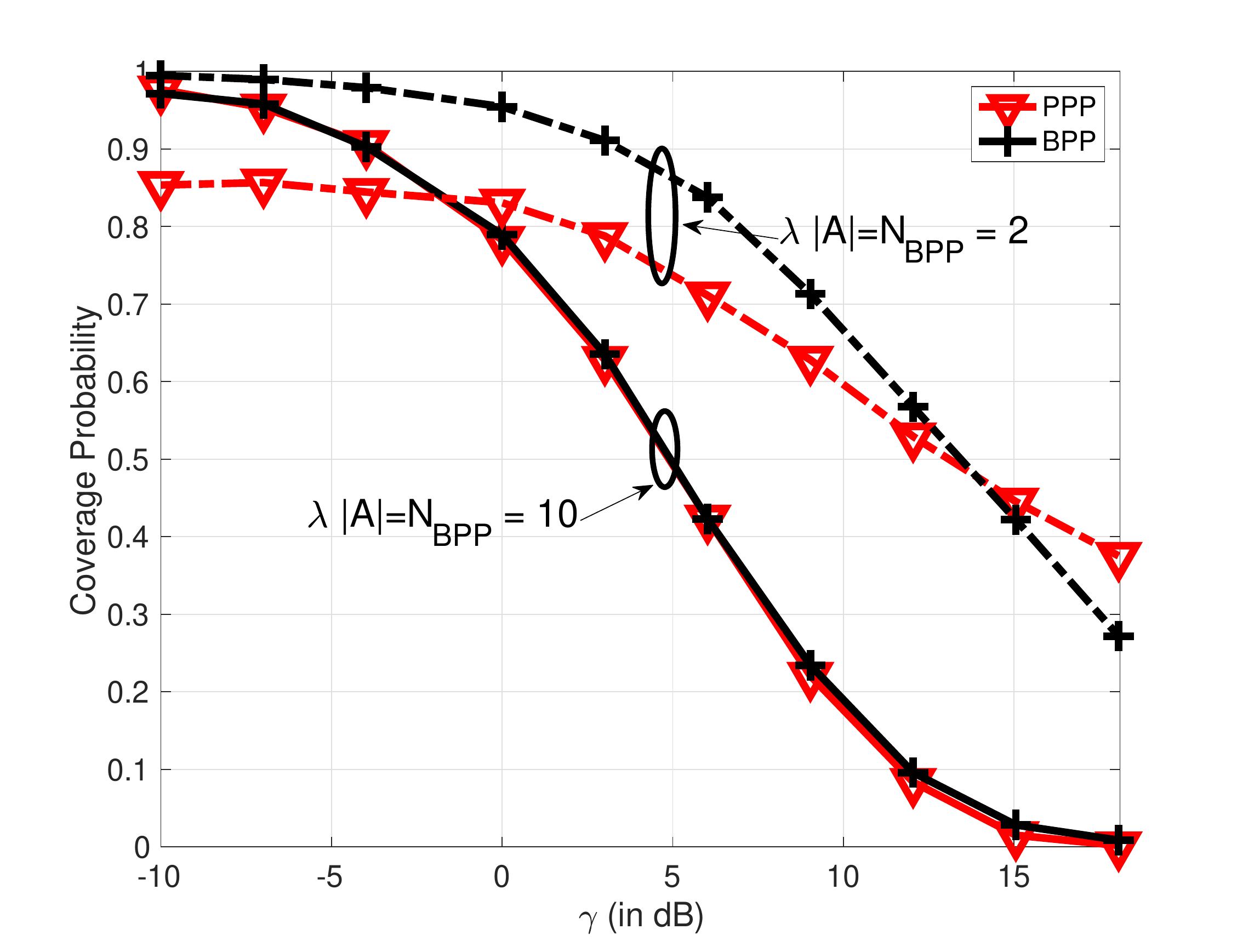}
    \caption{The coverage probabilities of PPP and BPP when $\alpha =2$, $h=R_{\sf S}-R_{\sf E}=550$ km, and $m=1$.  }\label{Fig_ppp_bpp}
\end{figure}

{\color{black}{
{\bf Comparison to realistic satellite constellations:} 
We compare the coverage probabilities obtained with a PPP model and actual Starlink satellite constellations to verify how well a PPP model captures realistic satellite constellations. Specifically, we collected $200$ random snapshot samples of visible satellites from Seoul, South Korea, from the website: https://satellitemap.space. Then, leveraging this data set, we composed downlink satellite communication networks where the typical user is located in Seoul. In such a satellite network, we assumed to use $20$ orthogonal frequency-time resources as in \cite{okati:tcom:20}, in which each satellite randomly selects one orthogonal resource. We note that this is equivalent to independent thinning in a PPP. Fig.~\ref{Fig_histo}-(a) depicts the empirical distribution for the number of visible satellites seen by the downlink user and the associated Poisson distribution with the moment matching, i.e., the estimated mean value is about $8.19$. Interestingly, the empirical distribution obtained from real data sets is tightly matched with the Poisson distribution. This result confirms that the number of satellites visible in the typical spherical cap $\mathcal{A}$ can be tightly approximated as the Poisson distribution. 
{\color{black}{Fig.~\ref{Fig_histo}-(b) shows the satellite spatial distribution in a single snapshot. 
}}

{\color{black}{
We clarify that our model captures effects on satellite constellations in a spatially average sense. To be specific, we model a satellite network by a PPP, so that multiple constellation snapshots are incorporated into one single expression; by which our results get analytical generality. As shown in Fig.~\ref{Fig_histo} and \ref{Fig_starlink}, this approach is well fitted in densely deployed satellite networks such as Starlink. Nonetheless, our approach is limited in that it cannot reflect specific conditions of satellite networks, e.g., particular satellite orbits. For this reason, our analysis may not accurately model orbit-dependent satellite networks such as OneWeb. To properly study such orbit-dependent satellite networks, we need a new model that can reflect orbit-specific characteristics into analysis as in \cite{lee:arxiv:22}; yet this model cannot have analytical tractability as our current analysis. 
}}


In Fig.~\ref{Fig_starlink}-(a), we present the coverage probabilities of PPP, BPP, and actual Starlink constellations by choosing {\color{black}{the minimum elevation angle $\psi_{\rm min} = 25^{\circ}$, which follows the grant from the FCC.}} The figure shows that the coverage probability of the actual Starlink constellations is well captured by a PPP model, justifying using a homogeneous PPP to model actual satellite constellations. 
Nevertheless, a noticeable gap exists between the coverage probabilities when $\gamma$ is large. We conjecture that this gap arises because, in actual constellations, each satellite location is repulsive, refraining that two satellite points are very closely located, while a PPP cannot capture such correlation. Studying such a sophisticated point process is promising for future work. {\color{black}{With a PPP model used in this paper, we can heuristically tune the density parameter as $\lambda_{\sf tune} = 0.75\lambda$ to decrease the coverage probability gap as shown in Fig.~\ref{Fig_starlink}-(b).}}
}}

\begin{figure}
\centering
\begin{subfigure}{0.75\columnwidth}
  \includegraphics[width=1\columnwidth]{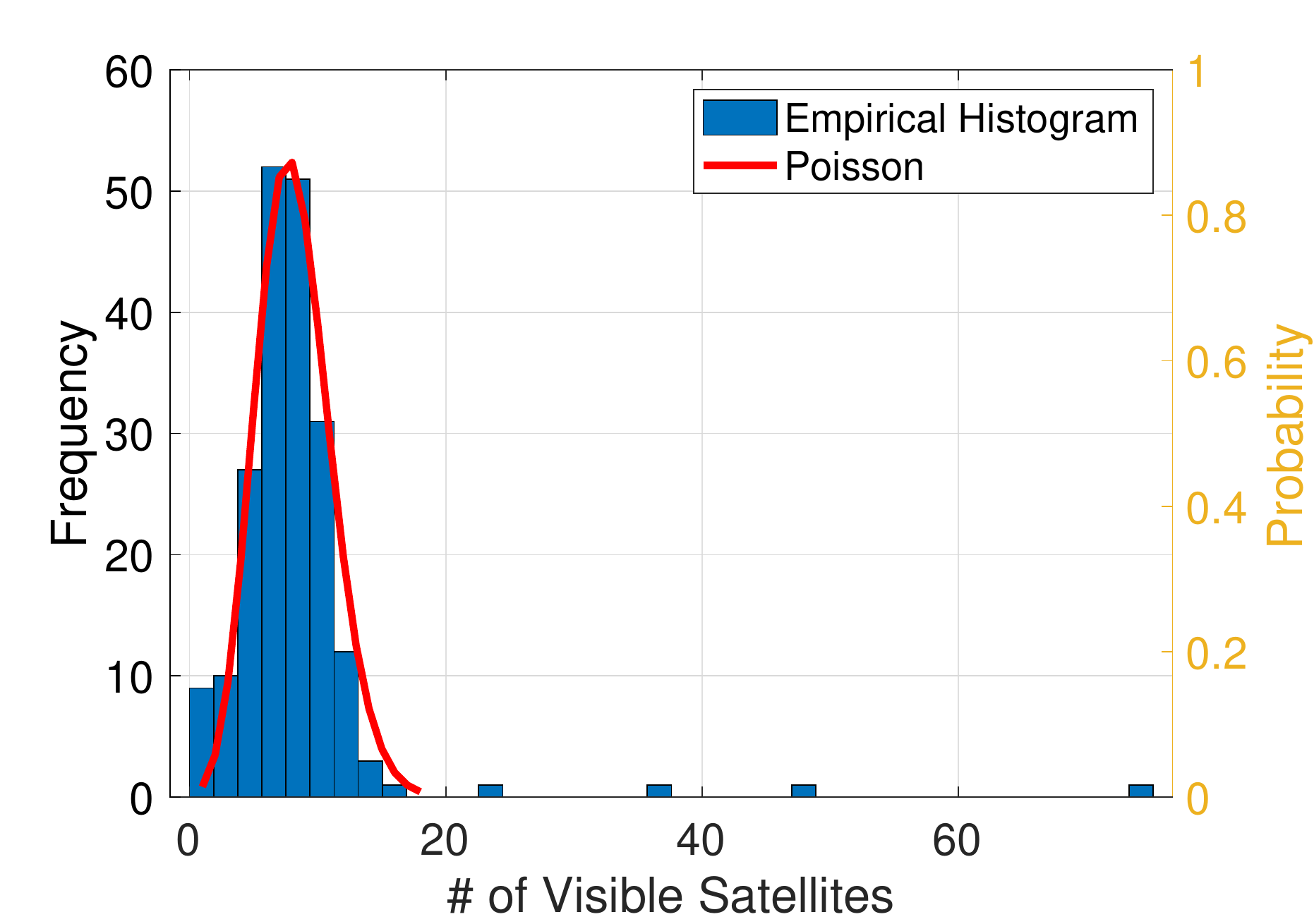}  
  \caption{}
\end{subfigure}
\begin{subfigure}{1\columnwidth}
  
  \includegraphics[width=1\columnwidth]{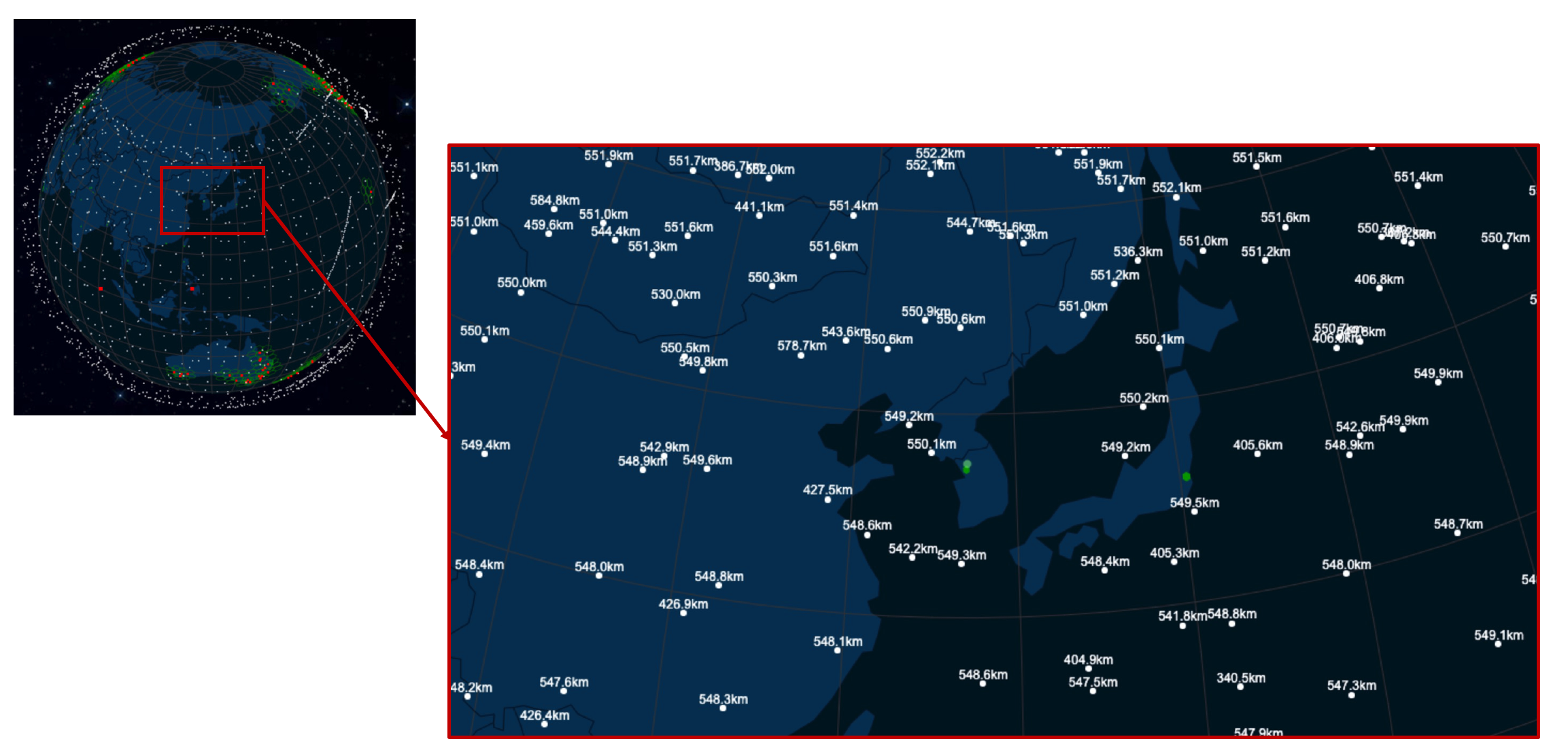}  
  \caption{}
\end{subfigure}
\caption{(a) The comparison between empirical histogram for the number of visible satellites for $200$ random snapshots vs. Poisson probability by matching the mean. The satellite constellation data comes from public tracking data published at http://satellitemap.space and http://space-track.org. (b) The satellite constellation map captured from http://satellitemap.space.}
\label{Fig_histo}
\end{figure}

\begin{figure}
\centering
\begin{subfigure}{.85\columnwidth}
  
  \includegraphics[width=1\columnwidth]{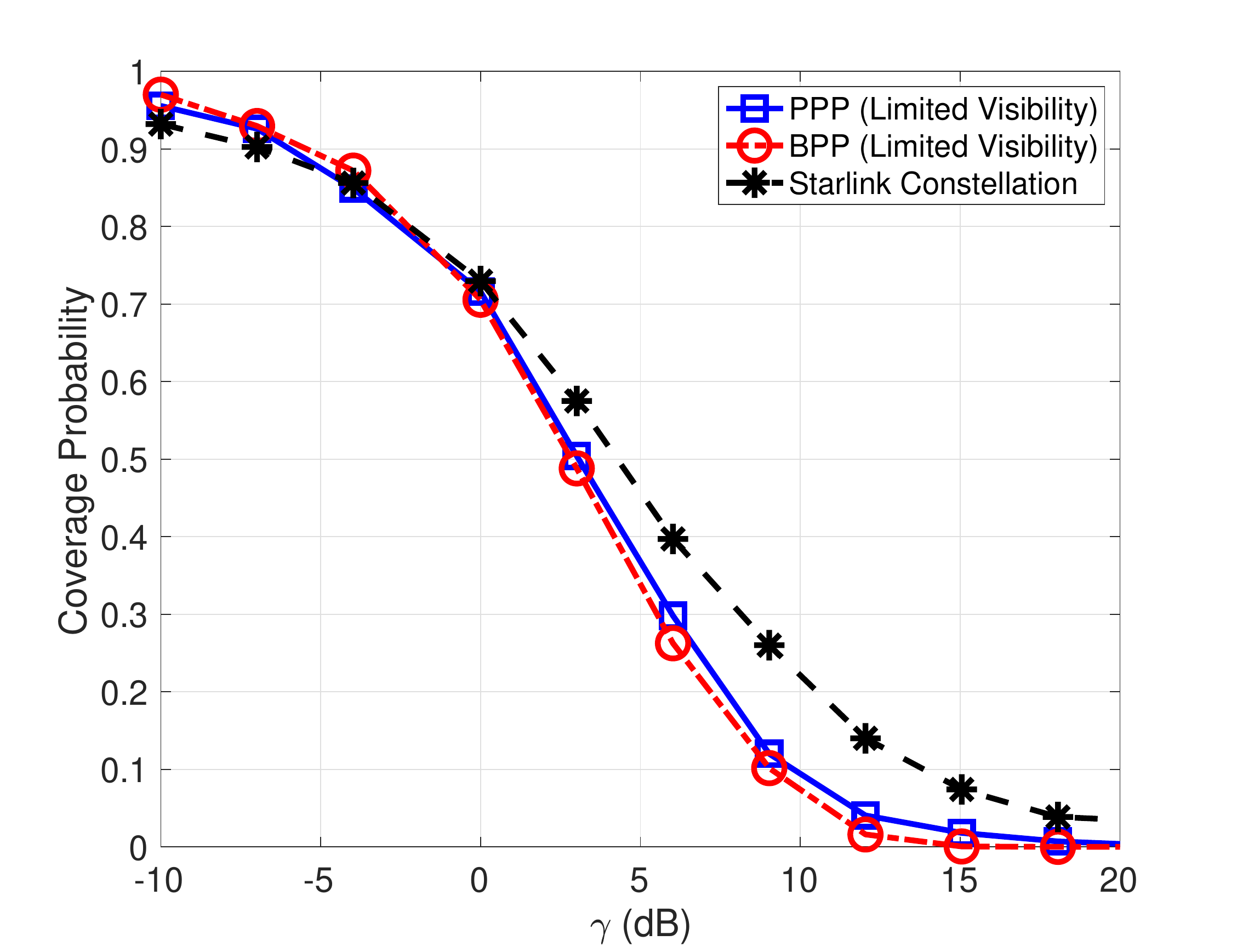}  
  \caption{}
\end{subfigure}
\begin{subfigure}{.85\columnwidth}
 
  \includegraphics[width=1\columnwidth]{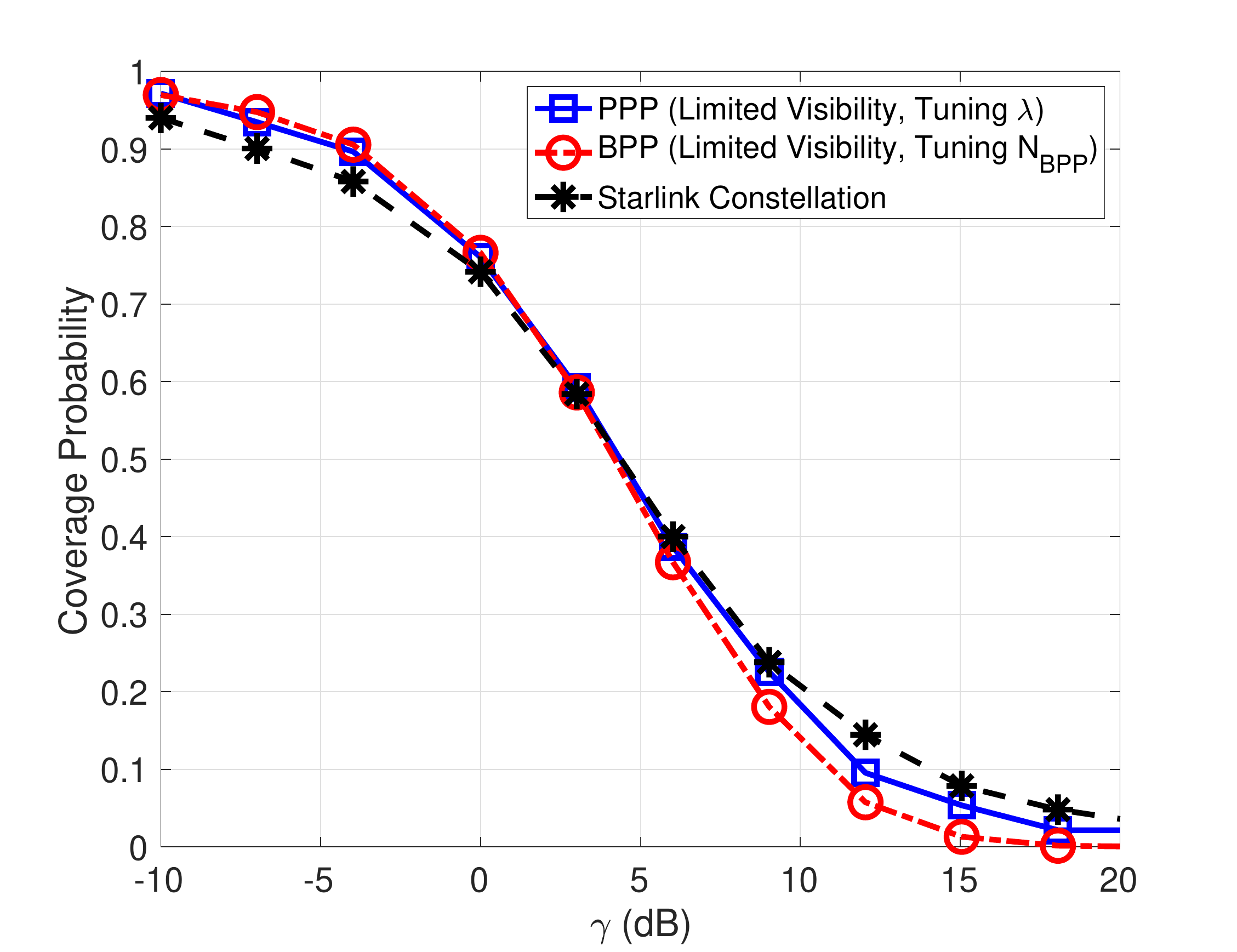}  
  \caption{}
\end{subfigure}
\caption{(a) The coverage probabilities of PPP, BPP, and actual Starlink constellations when $\alpha =2$, $h=R_{\sf S}-R_{\sf E}=550$ km, and $m=1$. (b) The coverage probabilities of PPP, BPP, and actual Starlink constellations by heuristically tuning the density parameter as $\lambda_{\sf tune} = 0.75\lambda$.}
\label{Fig_starlink}
\end{figure}



{\bf Effect of satellite density:} To see the effect of the satellite density, we evaluate the coverage probability by increasing the density of satellites for fixed $\alpha=2$, $h=R_{\sf S}-R_{\sf E}=500$ km, and $m=1$.  As can be seen in Fig.~\ref{Fig_combined}-(a), the coverage probability tends to enhance as decreasing the satellite density because of the reduced co-channel interference power. However, it is remarkable that the coverage performance declines, provided that the satellite density becomes smaller than $\lambda|\mathcal{A}|\simeq 1$. This phenomenon occurs because the probability of satellite-visibility exponentially decreases with the average number of satellites. Consequently, we need to carefully choose the density when deploying a satellite network to maximize the coverage performance.

\begin{figure}
\centering
\begin{subfigure}{.85\columnwidth}
 
  \includegraphics[width=1\columnwidth]{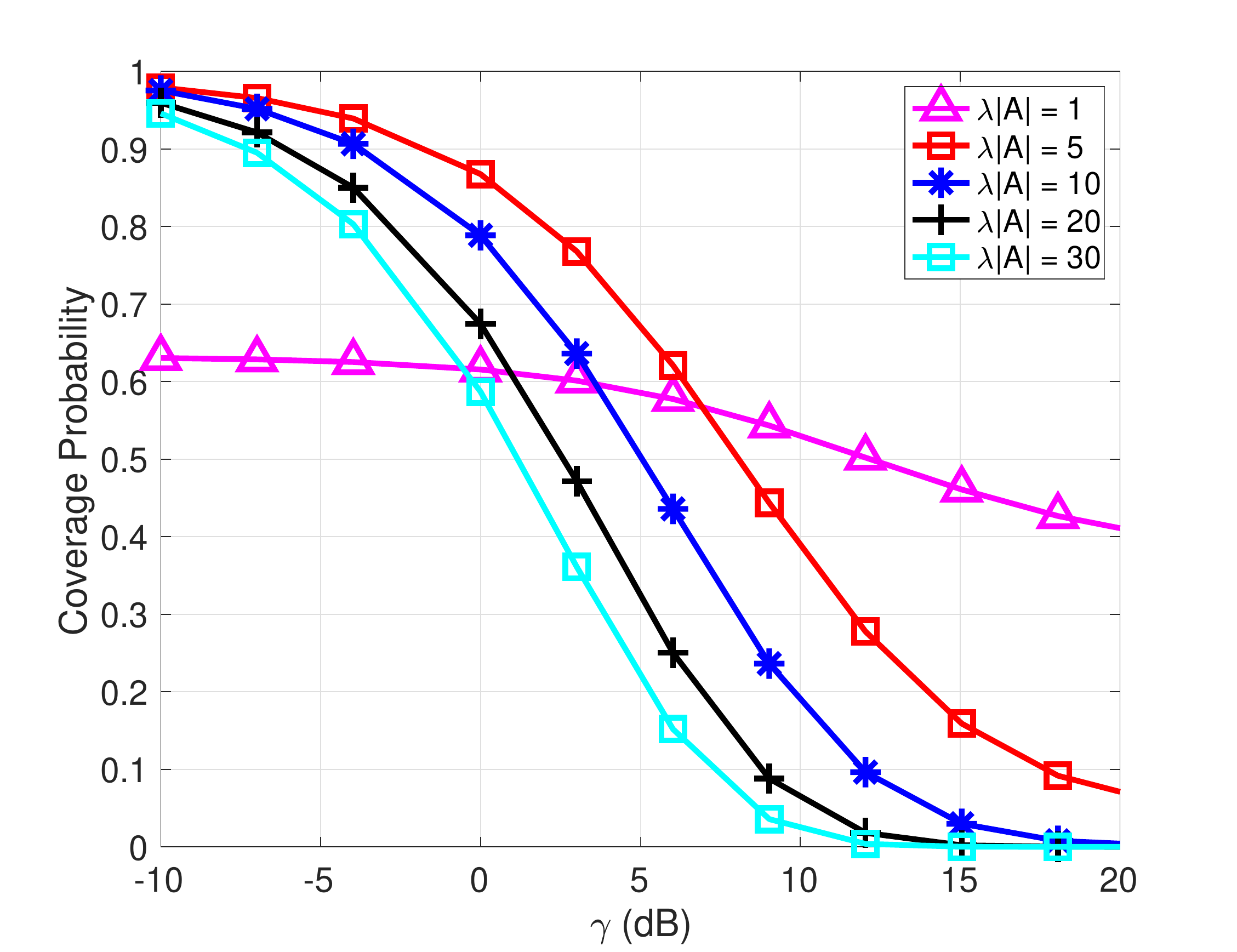}  
  \caption{}
\end{subfigure}
\begin{subfigure}{.85\columnwidth}
  \includegraphics[width=1\columnwidth]{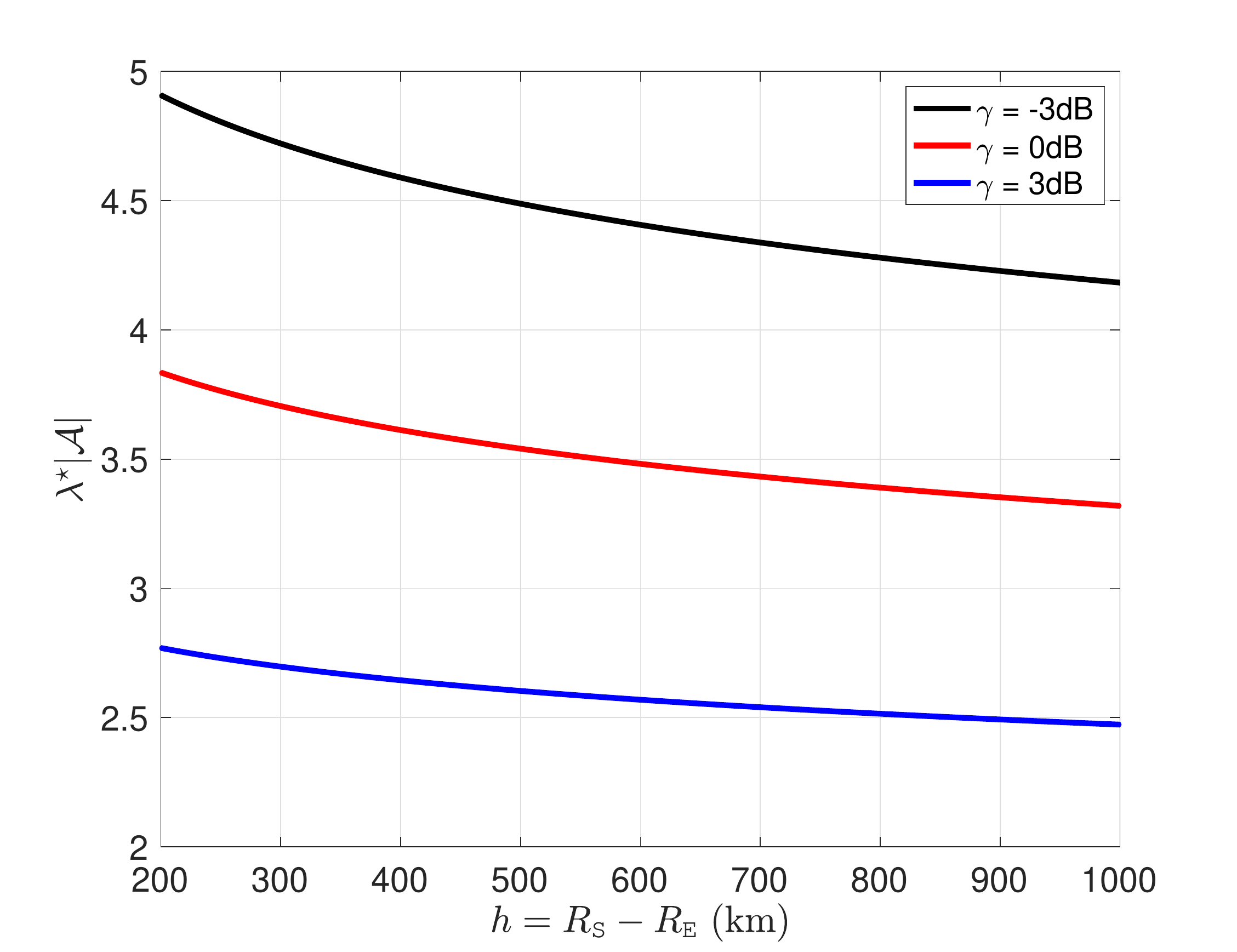}  
  \caption{}
\end{subfigure}
\caption{(a) The coverage probabilities per $\lambda |\mathcal{A}|$ when $\alpha =2$, $h=R_{\sf S}-R_{\sf E}=500$ km, and $m=1$. (b) The optimal average number of satellites versus the altitude of satellites  according to the target threshold $\gamma$ for fixed $\alpha =2$  and $m=1$.}
\label{Fig_combined}
\end{figure}



%

{\bf Optimal trade-off:} Fig.~\ref{Fig_combined}-(b) shows the optimal trade-off between the average number of satellites $\lambda^{\star}|\mathcal{A}|$ and the satellite altitude $h=R_{\sf S}-R_{\sf E}$ derived in \eqref{eq:trade-off}. As can be seen, the optimal average number scales down logarithmically with the altitude. This trade-off suggests that a network operator must carefully choose a satellite density depending on the network altitude to maximize the network coverage performance. 
{\color{black}{We note that the optimal satellite density is obtained by using tools of stochastic geometry, so that this is valid in a spatial average sense. 
}}


\section{Conclusions}
In this paper, we have characterized the coverage performance of satellite networks. Using stochastic geometry tools, we have derived the coverage probability in terms of the satellite density, the Nakagami fading parameter, and the path-loss exponent. From these obtained analytical expressions, we have shown that the LOS propagation is beneficial to reduce the fading effects, while it is less preferred due to high co-channel interference effects. More importantly, we have found that the optimal mean number of satellites for maximizing the coverage probability, and have shown that the optimal number decreases logarithmically with  the satellite altitude.   

Potential follow-up research directions include extending the analysis by considering the distance-dependent path-loss model and the beam misalignment effects, incorporating distinct LOS and NLOS probabilities, which offers more realistic coverage performance. It is also interesting to extend the coverage analysis for multi-tier heterogeneous satellite networks with multi-altitudes, and distinct transmit power levels per network tier.  Another worthwhile direction is the coverage analysis using Poisson line processes, which may better capture the satellite movements on an orbit. 
{\color{black}{Incorporating non-uniform traffics \cite{singh:twc:13} on the earth or more sophisticated beam gain models \cite{jia:iotj:22} are also promising. 
}}

\appendices

\section{Proof of Lemma \ref{lem2}} \label{proof:lem2}
\vspace{-0.1cm}
 We compute the probability that the nearest neighbor of a point ${\bf x}_1\in \Phi$ is larger than $r$ conditioned that at least more than one satellite exists in $\mathcal{A}$. To this end, we calculate the probability that there is no point in $\mathcal{A}_r$ as illustrated in Fig. \ref{fig:geo}, i.e., $\Phi(\mathcal{A}_r)=0$, conditioned $\Phi(\mathcal{A})>0$. Using Bayes' theorem, we compute
 \begin{align}
	& \mathbb{P}\left[R>r| \Phi(\mathcal{A})>0\right] \nonumber \\
	&=\mathbb{P}[\Phi(\mathcal{A}_r)=0|\Phi(\mathcal{A})>0] = \frac{\mathbb{P}[\Phi(\mathcal{A}_r)=0,\Phi(\mathcal{A})>0]}{ \mathbb{P}[\Phi(\mathcal{A})>0]} \nonumber\\
	&\mathop{=}^{(a)} \frac{\mathbb{P}[\Phi(\mathcal{A}_r)=0]\mathbb{P}[\Phi(\mathcal{A}/\mathcal{A}_r)>0]}{ \mathbb{P}[\Phi(\mathcal{A})>0]}\nonumber\\
	&= \frac{\mathbb{P}[\Phi(\mathcal{A}_r)\!=\!0](1\!-\!\mathbb{P}[\Phi(\mathcal{A}/\mathcal{A}_{r})\!=\!0])}{1\!-\!\mathbb{P}[\Phi(\mathcal{A})\!=\!0]} \nonumber\\
	&\mathop{=}^{(b)}\frac{\!\exp\left(-\lambda|\mathcal{A}_r|\right)(1\!-\!\exp\left(-\lambda |\mathcal{A}/\mathcal{A}_r|\!\right))}{1\!-\!\exp\left(\!-\lambda |\mathcal{A}|\!\right)}\nonumber\\
	&=\frac{\exp\left(-\lambda|\mathcal{A}_r|\right) -\exp\left(-\lambda |\mathcal{A}|\right)}{1-\exp\left(-\lambda |\mathcal{A}|\right)}, \label{eq:ccdf_ND}
\end{align}
where (a) follows from the independence of the PPP for non-overlapping areas $\mathcal{A}/\mathcal{A}_r$ and $\mathcal{A}_r$ and (b) is by Lemma 1. To accomplish this, we need to compute the area for $\mathcal{A}_r$. By the Archimedes' Hat-Box Theorem, the area of $\mathcal{A}_r$ is given by
 \begin{align}
 	|\mathcal{A}_r| = 2\pi (R_{\sf S}-R_{\sf E}-h_r)R_{\sf S}. \label{eq:arear}
 \end{align}
 As depicted in Fig. \ref{fig:geo}, we can express $h_r$ in terms of $r$ as
 \begin{align}
 	 h_r = \frac{(R_{\sf S}^2-R_{\sf E}^2)-r^2}{2R_{\sf E}}.\label{eq:hr}
 	  \end{align} 
 Invoking \eqref{eq:hr} into \eqref{eq:arear}, the area of $\mathcal{A}_r$ is represented as \begin{align}
 	|\mathcal{A}_r| = 2\pi \left(R_{\sf S}-R_{\sf E}-\frac{(R_{\sf S}^2-R_{\sf E}^2)-r^2}{2R_{\sf E}}\right)R_{\sf S}.
 \end{align}
As a result, the probability that no satellite exists in $\mathcal{A}_r$ is given by
 \begin{align}
 	&\mathbb{P}\left[\Phi(\mathcal{A}_r)=0\right]=e^{-\lambda |\mathcal{A}_r|} 
 	\nonumber \\
 	&=e^{-2\lambda \pi R_{\sf S}(R_{\sf S}-R_{\sf E})}e^{-\lambda \pi \frac{R_{\sf S}}{R_{\sf E}}\left\{r^2-(R_{\sf S}^2-R_{\sf E}^2))\right\}}. \label{eq:void1}
 \end{align}
 Also, the probability that there is no satellite in $\mathcal{A}$ is computed as
 \begin{align}
 		\mathbb{P}\left[\Phi(\mathcal{A})=0\right]
 		&=e^{-2\lambda \pi R_{\sf S}(R_{\sf S}-R_{\sf E})}.\label{eq:void2}
 \end{align}
 Plugging \eqref{eq:void1} and \eqref{eq:void2} into  \eqref{eq:ccdf_ND}, we obtain the conditional CCDF of $R$
 \begin{align}
 	F_{R|\Phi(\mathcal{A})>0}^c(r)
 	= \frac{e^{-2\lambda \pi R_{\sf S}(R_{\sf S}-R_{\sf E})}\left[e^{-\lambda \pi \frac{R_{\sf S}}{R_{\sf E}}\left\{r^2-(R_{\sf S}^2-R_{\sf E}^2)\right\}} -1\right]}{1-e^{-2\lambda \pi R_{\sf S}(R_{\sf S}-R_{\sf E})}}. 	\end{align}
 By taking derivative with respect to $r$, we obtain the conditional distribution of the nearest satellite distance as
 \begin{align}
 	&f_{R|\Phi(\mathcal{A})>0}(r) = \frac{\partial (1-F_{R|\Phi(\mathcal{A})>0}^c(r))}{\partial r}
\nonumber \\
& = 2\pi \lambda\frac{R_{\sf S}}{R_{\sf E}}  \frac{e^{\lambda \pi \frac{R_{\sf S}}{R_{\sf E}} (R_{\sf S}^2-R_{\sf E}^2) }}{e^{2\lambda \pi R_{\sf S}(R_{\sf S}-R_{\sf E})}-1} re^{-\lambda \pi \frac{R_{\sf S}}{R_{\sf E}}r^2},
 \end{align}
 where $R_{\rm min}\leq r \leq R_{\rm max}$.

\section{Proof of Lemma \ref{lem3}} \label{proof:lem3}
We first define the aggregated interference power conditioned that $\Phi(\mathcal{A})>0$ and the nearest satellite distance from the typical receiver is placed at fixed distance $\|{\bf x}_{1}-{\bf u}_1\|=r$ as 
$I_{r}=\sum_{{\bf x}_i\in \Phi \cap \mathcal{A}_r^c} {\bar G}_iH_i\|{\bf x}_i-{\bf u}_1\|^{-\alpha}$,
where $\mathcal{A}_r^c=\mathcal{A}\setminus \mathcal{A}_r$ denotes the area on the spherical cap outside $\mathcal{A}_r \subset\mathcal{A}$.  We compute the Laplace transform of the aggregated interference power. Such Laplace transform is computed as
\begin{align}
& \mathcal{L}_{I_{r}|\Phi(\mathcal{A})>0}(s) =\mathbb{E}\left[e^{-s I_r} ~ \middle|~ \|{\bf x}_{1}-{\bf u}_1\|=r,  \Phi(\mathcal{A})>0 \right] \label{eq:condLapla} \\
&\stackrel{(a)}{=} \exp\left( -\lambda \int_{v\in \mathcal{A}_r^c}  \left(1-\mathbb{E}\left[e^{-s {\bar G}_iH_iv^{-\alpha}}\right]\right)  ~  {\rm d}v \right)\nonumber\\
&\stackrel{(b)}{=} \exp\left( -\lambda \int_{v\in \mathcal{A}_r^c}  1- \frac{1}{\left(1+\frac{s{\bar G}_iv^{-\alpha}}{m}\right)^{m}}    ~  {\rm d}v \right) 
 \nonumber \\
&\stackrel{(c)}{=} \exp\left(-2\pi \lambda \frac{R_{\sf S}}{R_{\sf E}}\int_{r}^{R_{\rm max} } \left[ 1-   \frac{1}{ \left(1+\frac{s{\bar G}_iv^{-\alpha}}{m}\right)^{m}} \right] v~ {\rm d}v \right) \nonumber\\
&\stackrel{(d)}{=}  \exp    \left(  - \lambda   \pi  \frac{R_{\sf S}}{R_{\sf E}}       \left(   \frac{{\bar G}_is}{m}  \right)^{  \frac{2}{\alpha}}         \int_{\left(   \frac{{\bar G}_is}{m}  \right)^{-\frac{2}{\alpha}}r^2 }^{\left(   \frac{{\bar G}_is}{m}  \right)^{-\frac{2}{\alpha}}R_{\rm max}^2  } 1  -   \frac{1}{\left(1   +    u^{  -\frac{\alpha}{2}}   \right)^m}{\rm d} u       \right)    , \label{eq:LaplaRay}
\end{align}
where (a) follows from the probability generating functional (PGFL) of the PPP \cite{baccelli:tit:06,haenggi:jsac:09}, (b) is because {\color{black}$\sqrt{H_i}$ follows the Nakagami $m$ distribution}, (c) comes from $	\frac{\partial |\mathcal{A}_r|}{\partial r}=2\frac{R_{\sf S}}{R_{\sf E}}\pi r$, and (d) is the change of variable $u=\left(\frac{s{\bar G}_i}{m}\right)^{-\frac{2}{\alpha}}v^2$ {\color{black}{ and ${\rm d}u = 2v \left(\frac{s{\bar G}_i}{m}\right)^{-\frac{2}{\alpha}} {\rm d} v$ accordingly.}} 
Evaluating at $s=r^{\alpha}{  \gamma} $, we obtain 
\begin{align}
 \mathcal{L}_{I_{r}|\Phi(\mathcal{A})>0}(r^{\alpha} \gamma)= 	  \exp  \left(    - \lambda   \pi  \frac{R_{\sf S}}{R_{\sf E}} r^2    \eta(\lambda, \alpha, R_{\sf S}, m ; r,\gamma)     \right),  \label{eq:Laplace}
\end{align}
with 
\begin{align}
 \eta(\lambda, \alpha, R_{\sf S}, m ;r, \gamma) \triangleq	 \left(\!\frac{{\bar G_i}\gamma}{m}\!\right)^{\!\frac{2}{\alpha}}\int_{\left(\!\frac{{\bar G}_i{ \gamma}}{m}\!\right)^{\!\!-\frac{2}{\alpha}} }^{\left(\!\frac{{\bar G}_i{ \gamma}}{m}\!\right)^{\!\!-\frac{2}{\alpha}}\!\left(\!\frac{R_{\rm max}}{r}\!\right)^2  }\!\! 1\!-\! \frac{1}{\left(1 \!+\!  u^{-\frac{\alpha}{2}} \right)^m}{\rm d} u. \nonumber\\
\end{align}
This completes the proof.

 \section{Proof of Theorem \ref{Th1}}
\label{proof:Th1}
{\color{black}{
We first obtain the coverage probability as 
\begin{align}
    &P^{ {\sf cov}}_{{\sf SIR}} (\gamma; \lambda, \alpha,R_{\sf S},m) \nonumber \\
    & = \mathbb{P} [\Phi(\mathcal{A}) = 0] \cdot P^{ {\sf cov}}_{{\sf SIR} | \Phi(\mathcal{A}) = 0} (\gamma; \lambda, \alpha,R_{\sf S},m) + \nonumber \\
    & \mathbb{P} [\Phi(\mathcal{A}) > 0] \cdot P^{ {\sf cov}}_{{\sf SIR} | \Phi(\mathcal{A}) > 0} (\gamma; \lambda, \alpha,R_{\sf S},m).
\end{align}
Since $P^{ {\sf cov}}_{{\sf SIR} | \Phi(\mathcal{A}) = 0} (\gamma; \lambda, \alpha,R_{\sf S},m) = 0$,}} conditioning on $\|{\bf x}_{1}-{\bf u}_1\|=r $ and $\Phi(\mathcal{A})>0$, we compute as 
\begin{align}
&P^{\sf cov}_{{\sf SIR}|\Phi(\mathcal{A})>0} (\gamma;\lambda, \alpha,R_{\sf S}, m) = \nonumber \\
& \mathbb{E}\left[\mathbb{P}\left[  H_1 \geq  {r}^{\alpha}\gamma  I_r  \mid   \Phi(\mathcal{A})>0 , \|{\bf x}_{1}-{\bf u}_1\|=r \right] \mid  \Phi(\mathcal{A})>0 \right], \label{eq:condprob1}
\end{align}
where the expectation is taken over the distribution of $r$ conditioned on $\Phi(\mathcal{A})>0$. Since {\color{black} $\sqrt{H_i}$} is the Nakagami-$m$ random variable, the CCDF of $H_i$ is given by
\begin{align} \label{eq:ccdf_H}
	 \mathbb{P}[H_i\geq x] 	&=  e^{-m x}\sum_{k=0}^{m-1}\frac{(m x)^k}{k!}.
\end{align}
Utilizing this, we can rewrite the coverage probability in \eqref{eq:condprob1} as
\begin{align}
 &P^{ {\sf cov}}_{{\sf SIR}|\Phi(\mathcal{A})>0} (\gamma; \lambda, \alpha,R_{\sf S},m) \nonumber\\
&=\mathbb{E}\Bigg[  \mathbb{E} \Bigg[  \sum_{k=0}^{ m  -1} \frac{ m^k\gamma^k{r}^{\alpha k}}{k!}{ I}_r^k  e^{-m  {r}^{\alpha}\gamma { I}_r  }  \nonumber \\
& \mid  \Phi(\mathcal{A}) > 0 , \|{\bf x}_{1}-{\bf u}_1\|=r \Bigg]  \mid \Phi(\mathcal{A})>0 \Bigg]  \nonumber \\
&\stackrel{(a)}{=} \! \mathbb{E}\left[ \sum_{k=0}^{ m -1}\frac{ m^k \gamma^k {r}^{\alpha k}}{k!} (-1)^{k}\frac{\d^k\mathcal{L}_{{ I}_{r|\Phi(\mathcal{A})>0} } (s)}{\d s^k}  \left|_{s= m\gamma {r}^{\alpha}} \right.   \! \mid\! \Phi(\mathcal{A})>0\right]\nonumber\\
&\stackrel{(b)}{=}\nu(\lambda, R_{\sf S})\!\!\int_{R_{\rm min}}^{R_{\rm max}}  \sum_{k=0}^{ m -1}\frac{ m^k\gamma^k {r}^{\alpha k}}{k!} (-1)^{k}\left.\!{\frac{\d^k\mathcal{L}_{{  I}_{r|\Phi(\mathcal{A})>0} }(s)}{\d s^k}} \right|_{s= m\gamma  {r}^{\alpha}} \nonumber \\
& re^{-\lambda \pi \frac{R_{\sf S}}{R_{\sf E}}r^2} {\rm d} r,
\end{align}
where (a) follows from Lemma \ref{lem3} and applying the derivative property of the Laplace transform, i.e., $\mathbb{E}\left[X^{k} e^{-sX}\right]=(-1)^{k}\frac{\d^k\mathcal{L}_X(s)}{\d s^k}$, and (b) is due to the expectation over the nearest distribution given in Lemma \ref{lem2}. By multiplying $\mathbb{P}[\Phi(\mathcal{A})>0]= 1-e^{-\lambda 2\pi(R_{\sf S}-R_{\sf E})R_{\sf S}  }$, we finally obtain the expression in \eqref{eq:Th1}, which completes the proof.

\section{Proof of Theorem \ref{Th2}}
\label{proof:Th2}

 We devote to proving the upper bound, since the lower bound is readily obtained from the former by choosing $\kappa=1$. Conditioned on the nearest satellite being placed at distance $r$ from the typical receiver's location, the conditional probability can be written as
\begin{align}
&P^{\sf cov}_{{\sf SIR}|\Phi(\mathcal{A})>0} (\gamma;\lambda, \alpha,R_{\sf S}, m) = \nonumber \\
& \mathbb{E}\left[\mathbb{P}\left[  H_1 \geq  {r}^{\alpha}\gamma  I_r  \mid   \Phi(\mathcal{A})>0 , R_1=r \right] \mid  \Phi(\mathcal{A})>0 \right].\label{eq:condprob2}
\end{align}
Recall that the CCDF for $H_1$ can be represented in terms of  the lower incomplete gamma function as
\begin{align}
\mathbb{P}[H_1 > x] 
&= 1- \frac{1}{\Gamma(m)}\int_{0}^{mx} t^{m-1}e^{-t}{\rm d}t.\label{eq:ccdf_H0}
\end{align}
From the Alzer's inequality \cite{alzer:97,lee:twc:15},  the incomplete Gamma function has an expression in the middle sandwiched between two inequalities: 
\begin{align}
	   \left(1-e^{-m \kappa x}\right)^{m}
\leq  \frac{1}{\Gamma(m)}\int_{0}^{mx} t^{m-1}e^{-t}{\rm d}t  \leq     \left(1-e^{-m  x}\right)^{m}.
\end{align}
Using this, the CCDF for $H_1$ is upper and lower bounded by
\begin{align}
  \left(1-e^{-m  x}\right)^{m}
	\leq \mathbb{P}[H_1 > x] \leq   \left(1-e^{-m \kappa x}\right)^{m},
\end{align}
where $\kappa=(m!)^{-\frac{1}{m}}$ and the equality holds when $m=1$.  Applying  the binomial expansion
\begin{align}
	 1-   \left(1-e^{-m \kappa x}\right)^{m} = \sum_{\ell=1}^{m} \binom{m}{\ell}(-1)^{\ell+1 }e^{-  \ell m \kappa x}, \nonumber
\end{align}
and plugging \eqref{eq:ccdf_H0} into \eqref{eq:condprob2}, we obtain an upper bound of the conditional coverage probability as
\begin{align}
&P^{\sf cov}_{{\sf SIR}|\Phi(\mathcal{A})>0} (\gamma;\lambda, \alpha,R_{\sf S}, m) \nonumber \\
& \leq  \sum_{\ell=1}^{m} \binom{m}{\ell}(-1)^{\ell+1 }   \mathbb{E}\left[ \mathbb{E}\left[ e^{-\ell m \kappa r^{\alpha}\gamma I_r } \mid R_1=r, \Phi(\mathcal{A})>0\right] \right. \nonumber \\
& \left. \mid \Phi(\mathcal{A})>0\right]\nonumber\\
&=\sum_{\ell=1}^{m} \binom{m}{\ell}(-1)^{\ell+1 } \mathbb{E}\left[  \mathcal{L}_{I_{r}|\Phi(\mathcal{A})>0}\left(\ell m \kappa  r^{\alpha}\gamma \right)  \! \mid \Phi(\mathcal{A})>0\right],\label{eq:Cov_prob_v1}
\end{align}
where the remaining expectation is taken over the distribution of    $f_{R|\Phi(\mathcal{A})}(r)$ in Lemma \ref{lem2}.  Toward this end, we compute the conditional Laplace transform of the aggregated interference power. From the Lemma \ref{lem3} in \eqref{eq:LaplaRay}, such Laplace transform is 
 \begin{align}
 & \mathcal{L}_{I_{r}|\Phi(\mathcal{A})>0}(\ell m  \kappa r^{\alpha}\gamma)= \nonumber \\
 &  \exp \left(  - \lambda   \pi  \frac{R_{\sf S}}{R_{\sf E}} r^2  \eta(\lambda, \alpha, R_{\sf S}, m ; r,\ell m \kappa \gamma)   \right). \label{eq:Laplace}
\end{align}
 Invoking \eqref{eq:Laplace} into \eqref{eq:Cov_prob_v1} and computing the expectation with respect to the nearest distance distribution given in Lemma \ref{lem2}, we obtain $P^{{\sf{ cov}}, {\rm B}}_{{\sf SIR}|\Phi(\mathcal{A})>0} (\gamma;\lambda, \alpha,R_{\sf S}, m)$. Multiplying  $\mathbb{P}[\Phi(\mathcal{A})>0]= 1-e^{-\lambda 2\pi(R_{\sf S}-R_{\sf E})R_{\sf S}  }$ to the conditional coverage probability, we arrive at the expression in \eqref{eq:Th2}, which completes the proof.

\section{Proof of Corollary \ref{Cor2}}
\label{proof:cor2}
  Using the fact that $\frac{R_{\rm max}}{r} \leq \frac{R_{\rm \max}}{R_{\rm min}}$ for $R_{\rm min}<r<R_{\rm max}$, we first define an upper bound of  $\eta^{\sf U}(x; \alpha, R_{\sf S}, m)$ in \eqref{eq:eta}, which is independent of $r$, as
\begin{align}
 \eta^{\sf U}(x; \alpha, R_{\sf S}, m) =\left( \frac{{\bar G_i}x}{m}\!\right)^{\!\frac{2}{\alpha}}\int_{\left(\!\frac{{\bar G}_i{ x}}{m}\!\right)^{\!\!-\frac{2}{\alpha}} }^{\left(\!\frac{{\bar G}_i{ x}}{m}\!\right)^{\!\!-\frac{2}{\alpha}}\!\left(\!\frac{R_{\rm max}}{R_{\rm min}}\!\right)^2  }\!\! 1\!-\! \frac{1}{\left(1 \!+\!  u^{-\frac{\alpha}{2}} \right)^m}{\rm d} u.\label{eq:integralupper}
\end{align}
Invoking \eqref{eq:integralupper} into \eqref{eq:Laplace} and evaluating at $x=\ell m  r^{\alpha}\gamma$, we obtain a lower bound of the conditional Laplace transform of the aggregated interference power  as
\begin{align}
 \mathcal{L}_{I_{r}|\Phi(\mathcal{A})>0}(\ell m  r^{\alpha}\gamma) \geq	 \exp\!\left( - \lambda   \pi  \frac{R_{\sf S}}{R_{\sf E}}     \eta^{\sf U}(\ell m \gamma  ; \alpha, R_{\sf S}, m) r^2  \!\right).\label{eq:Laplace_upper}
\end{align}
Using \eqref{eq:Laplace_upper}, we obtain a lower bound of $\mathbb{E}\left[  \mathcal{L}_{I_{r}|\Phi(\mathcal{A})>0}\left(\ell m    r^{\alpha}\gamma \right)   \! \mid \Phi(\mathcal{A})>0\right] \mathbb{P}[\Phi(\mathcal{A})>0]$ as
\begin{align}
&	 \mathbb{E}\left[  \mathcal{L}_{I_{r}|\Phi(\mathcal{A})>0}\left(\ell m    r^{\alpha}\gamma \right)   \! \mid \Phi(\mathcal{A})>0\right] \mathbb{P}[\Phi(\mathcal{A})>0]\nonumber\\
	&\geq \mathbb{E}\left[ e^{-\pi  \lambda \frac{R_{\sf S}}{R_{\sf E}} \eta^{\sf U}(\ell m \gamma  ; \alpha, R_{\sf S}, m) r^2 }  \mid \Phi(\mathcal{A})>0 \right]\mathbb{P}[\Phi(\mathcal{A})>0]\nonumber\\
	 &=e^{\lambda  \pi  \frac{R_{\sf S}}{R_{\sf E}} (R_{\sf S} - R_{\sf E})^2 } \left[ \frac{e^{-\lambda \pi \frac{R_{\sf S}}{R_{\sf E}} (1+\eta^{\sf U}(\ell m \gamma  ; \alpha, R_{\sf S}, m)) R_{\rm min}^2 }  }{1+\eta^{\sf U}(\ell m \gamma  ; \alpha, R_{\sf S}, m) } - \right. \nonumber \\
	 & \left.  \frac{ e^{-\lambda \pi \frac{R_{\sf S}}{R_{\sf E}}  ( 1 +\eta^{\sf U}(\ell m \gamma  ; \alpha, R_{\sf S}, m) ) R_{\rm max}^2 } }{1+\eta^{\sf U}(\ell m \gamma  ; \alpha, R_{\sf S}, m) }
	 \right].	\label{eq:Laplace_upper2}
\end{align}
 Invoking \eqref{eq:Laplace_upper2} into \eqref{eq:Cov_prob_v1}, we complete the proof. 

\section{Proof of Theorem \ref{Th3}}
\label{proof:Th3}

For notational simplicity, we let 
\begin{align}
    & a= \pi \frac{R_{\sf S}}{R_{\sf E}} \eta^{\sf U}(\gamma; \alpha, R_{\sf S}, 1 ) R_{\rm min}^2 , \\
    & b= \pi \frac{R_{\sf S}}{R_{\sf E}}\! \left[ \!\left(1\!+\!\eta^{\sf U}(\gamma; \alpha, R_{\sf S}, 1 )\!\right) \!R_{\rm max}^2\!-\!R_{\rm min}^2\!\right], \\ 
    & c=1+\eta^{\sf U}(\gamma; \alpha, R_{\sf S}, 1 ).
\end{align}
Then, the lower bound of the coverage probability in \eqref{eq:cor2}, can be expressed as
  \begin{align}
  	P^{\sf  cov,L}_{{\sf SIR}}  (\gamma ;\lambda, \alpha,R_{\sf S},1)=\frac{e^{-a \lambda } -e^{-b\lambda}}{c},\label{eq:simple}
  \end{align}
  where $0<a<b$ and $0<c$.  From \eqref{eq:simple}, we notice that $P^{\sf  cov,L}_{{\sf SIR}}  (\gamma ;\lambda, \alpha,R_{\sf S},1)$ is a neither convex nor concave function with respect to any positive value of $\lambda$. Therefore,  to verify the existence and the uniqueness of $\lambda^{\star}$, we first need to show that $P^{\sf  cov,L}_{{\sf SIR}}  (\gamma ;\lambda, \alpha,R_{\sf S},1)$ is unimodal function with respect to $\lambda>0$. Toward this end, a manipulation deduces that $P^{\sf  cov,L}_{{\sf SIR}}  (\gamma ;\lambda, \alpha,R_{\sf S},1)$ is increasing on $\lambda \in [0, \lambda^{\star}]$ and decreasing on $[\lambda^{\star}, \infty)$: 
    \begin{align}
  	&\frac{\partial P^{\sf  cov,L}_{{\sf SIR}}  (\gamma ;\lambda, \alpha,R_{\sf S},1)}{\partial \lambda} >0  \Leftrightarrow \frac{-ae^{-a \lambda } +be^{-b\lambda}}{c}>0  \nonumber \\
  	& \Leftrightarrow \lambda > \frac{\ln\left(\frac{b}{a}\right)}{b-a}=\lambda^{\star}. \label{eq:optimal}
  \end{align} 
By virtue of the unimodality of $P^{\sf  cov,L}_{{\sf SIR}}  (\gamma ;\lambda, \alpha,R_{\sf S},1)$, $P^{\sf  cov,L}_{{\sf SIR}}  (\gamma ;\lambda, \alpha,R_{\sf S},1)$ has a unique maximum at $\lambda =\lambda^{\star}$. Plugging the definitions of $a$, $b$, and $c$ into \eqref{eq:optimal}, the optimal density for given network parameters is given by
      	\begin{align}
   		\lambda^{\star} &= \frac{\ln\left(\frac{\left[1+\eta^{\sf U}(\gamma; \alpha, R_{\sf S}, 1 )\right]R_{\rm max}^2-R_{\rm min}^2}{ \eta^{\sf U}(\gamma; \alpha, R_{\sf S}, 1 ) R_{\rm min}^2}\right)}{  \pi \frac{R_{\sf S}}{R_{\sf E}}\left[1+\eta^{\sf U}(\gamma; \alpha, R_{\sf S}, 1 )\right]\left(R_{\rm max}^2-R_{\rm min}^2\right)} \nonumber \\
   		& = \frac{ \ln\left(1 + \frac{2 R_{\sf E}\left[1+\eta^{\sf U}(\gamma; \alpha, R_{\sf S}, 1 )\right]}{ \eta^{\sf U}(\gamma; \alpha, R_{\sf S}, 1 )R_{\rm min}}\right)}{\left[1+\eta^{\sf U}(\gamma; \alpha, R_{\sf S}, 1 )\right]  (2\pi R_{\sf S} R_{\rm min})}. 
   		   	\end{align}
Then we arrive at the expression in \eqref{eq:opt_density}, which completes the proof.

%





\bibliographystyle{IEEEtran}
\bibliography{ref_LEO_PPP}

\begin{thebibliography}{10}
\providecommand{\url}[1]{#1}
\csname url@samestyle\endcsname
\providecommand{\newblock}{\relax}
\providecommand{\bibinfo}[2]{#2}
\providecommand{\BIBentrySTDinterwordspacing}{\spaceskip=0pt\relax}
\providecommand{\BIBentryALTinterwordstretchfactor}{4}
\providecommand{\BIBentryALTinterwordspacing}{\spaceskip=\fontdimen2\font plus
\BIBentryALTinterwordstretchfactor\fontdimen3\font minus
  \fontdimen4\font\relax}
\providecommand{\BIBforeignlanguage}[2]{{%
\expandafter\ifx\csname l@#1\endcsname\relax
\typeout{** WARNING: IEEEtran.bst: No hyphenation pattern has been}%
\typeout{** loaded for the language `#1'. Using the pattern for}%
\typeout{** the default language instead.}%
\else
\language=\csname l@#1\endcsname
\fi
#2}}
\providecommand{\BIBdecl}{\relax}
\BIBdecl

\bibitem{liu:commmag:21}
S.~Liu, Z.~Gao, Y.~Wu, D.~W. Kwan~Ng, X.~Gao, K.-K. Wong, S.~Chatzinotas, and
  B.~Ottersten, ``{LEO} satellite constellations for {5G} and beyond: {How}
  will they reshape vertical domains?'' \emph{IEEE Commun. Mag.}, vol.~59,
  no.~7, pp. 30--36, 2021.

\bibitem{giordani:network:21}
M.~Giordani and M.~Zorzi, ``Non-terrestrial networks in the {6G} era:
  {Challenges} and opportunities,'' \emph{IEEE Network}, vol.~35, no.~2, pp.
  244--251, 2021.

\bibitem{gilbert:josiam:61}
E.~N. Gilbert, ``Random plane networks,'' \emph{Jour. of The Society for
  Industrial and Applied Mathematics}, vol.~9, pp. 533--543, 1961.

\bibitem{baccelli:book:09}
F.~Baccelli and B.~Blaszczyszyn, ``Stochastic geometry and wireless networks:
  {Volume} i theory,'' \emph{Found. Trends in Networking}, vol.~3, no. 3–4,
  p. 249–449, Mar. 2009.

\bibitem{haenggi:tit:08}
M.~Haenggi, ``A geometric interpretation of fading in wireless networks:
  {Theory} and applications,'' \emph{IEEE Trans. Inform. Theory}, vol.~54,
  no.~12, pp. 5500--5510, 2008.

\bibitem{baccelli:tit:06}
F.~Baccelli, B.~Blaszczyszyn, and P.~Muhlethaler, ``An {Aloha} protocol for
  multihop mobile wireless networks,'' \emph{IEEE Trans. Inform. Theory},
  vol.~52, no.~2, pp. 421--436, 2006.

\bibitem{baccelli:jsac:09}
------, ``Stochastic analysis of spatial and opportunistic {Aloha},''
  \emph{IEEE J. Sel. Areas Commun.}, vol.~27, no.~7, pp. 1105--1119, 2009.

\bibitem{haenggi:jsac:09}
M.~Haenggi, J.~G. Andrews, F.~Baccelli, O.~Dousse, and M.~Franceschetti,
  ``Stochastic geometry and random graphs for the analysis and design of
  wireless networks,'' \emph{IEEE J. Sel. Areas Commun.}, vol.~27, no.~7, pp.
  1029--1046, 2009.

\bibitem{andrews:tcom:11}
J.~G. Andrews, F.~Baccelli, and R.~K. Ganti, ``A tractable approach to coverage
  and rate in cellular networks,'' \emph{IEEE Trans. Commun.}, vol.~59, no.~11,
  pp. 3122--3134, 2011.

\bibitem{dhillon:jsac:12}
H.~S. Dhillon, R.~K. Ganti, F.~Baccelli, and J.~G. Andrews, ``Modeling and
  analysis of {K}-tier downlink heterogeneous cellular networks,'' \emph{IEEE
  J. Sel. Areas Commun.}, vol.~30, no.~3, pp. 550--560, 2012.

\bibitem{lee:jsac:15}
N.~Lee, X.~Lin, J.~G. Andrews, and R.~W. Heath, ``Power control for {D2D}
  underlaid cellular networks: {Modeling}, algorithms, and analysis,''
  \emph{IEEE J. Sel. Areas Commun.}, vol.~33, no.~1, pp. 1--13, 2015.

\bibitem{lee:twc:15}
N.~Lee, D.~Morales-Jimenez, A.~Lozano, and R.~W. Heath, ``Spectral efficiency
  of dynamic coordinated beamforming: {A} stochastic geometry approach,''
  \emph{IEEE Trans. Wireless Commun.}, vol.~14, no.~1, pp. 230--241, 2015.

\bibitem{huang:tit:13}
K.~Huang and J.~G. Andrews, ``An analytical framework for multicell cooperation
  via stochastic geometry and large deviations,'' \emph{IEEE Trans. Inform.
  Theory}, vol.~59, no.~4, pp. 2501--2516, 2013.

\bibitem{park:twc:16}
J.~{Park}, N.~{Lee}, J.~G. {Andrews}, and R.~W. {Heath}, ``On the optimal
  feedback rate in interference-limited multi-antenna cellular systems,''
  \emph{IEEE Trans. Wireless Commun.}, vol.~15, no.~8, pp. 5748--5762, 2016.

\bibitem{lee:tit:16}
N.~Lee, F.~Baccelli, and R.~W. Heath, ``Spectral efficiency scaling laws in
  dense random wireless networks with multiple receive antennas,'' \emph{IEEE
  Trans. Inform. Theory}, vol.~62, no.~3, pp. 1344--1359, 2016.

\bibitem{bai:commmag:14}
T.~Bai, A.~Alkhateeb, and R.~W. Heath, ``Coverage and capacity of
  millimeter-wave cellular networks,'' \emph{IEEE Commun. Mag.}, vol.~52,
  no.~9, pp. 70--77, 2014.

\bibitem{park:tccn:18}
J.~Park, J.~G. Andrews, and R.~W. Heath, ``Inter-operator base station
  coordination in spectrum-shared millimeter wave cellular networks,''
  \emph{IEEE Trans. Cognitive Commun. and Networking}, vol.~4, no.~3, pp.
  513--528, 2018.

\bibitem{renzo:twc:15}
M.~Di~Renzo, ``Stochastic geometry modeling and analysis of multi-tier
  millimeter wave cellular networks,'' \emph{IEEE Trans. Wireless Commun.},
  vol.~14, no.~9, pp. 5038--5057, 2015.

\bibitem{atzeni:twc:18}
I.~Atzeni, J.~Arnau, and M.~Kountouris, ``Downlink cellular network analysis
  with {LOS/NLOS} propagation and elevated base stations,'' \emph{IEEE Trans.
  Wireless Commun.}, vol.~17, no.~1, pp. 142--156, 2018.

\bibitem{chetlur:tcom:17}
V.~V. Chetlur and H.~S. Dhillon, ``Downlink coverage analysis for a finite
  {3-D} wireless network of unmanned aerial vehicles,'' \emph{IEEE Trans.
  Commun.}, vol.~65, no.~10, pp. 4543--4558, 2017.

\bibitem{banagar:twc:20}
M.~Banagar and H.~S. Dhillon, ``Performance characterization of canonical
  mobility models in drone cellular networks,'' \emph{IEEE Trans. Wireless
  Commun.}, vol.~19, no.~7, pp. 4994--5009, 2020.

\bibitem{ganz:tcom:94}
A.~Ganz, Y.~Gong, and B.~Li, ``Performance study of low {Earth}-orbit satellite
  systems,'' \emph{IEEE Trans. Commun.}, vol.~42, no. 234, pp. 1866--1871,
  1994.

\bibitem{vatalaro:jsac:95}
F.~Vatalaro, G.~Corazza, C.~Caini, and C.~Ferrarelli, ``Analysis of {LEO, MEO,
  and GEO} global mobile satellite systems in the presence of interference and
  fading,'' \emph{IEEE J. Sel. Areas Commun.}, vol.~13, no.~2, pp. 291--300,
  1995.

\bibitem{mokhtar:wcl:20}
A.~Mokhtar and M.~Azizoglu, ``On the downlink throughput of a broadband {LEO}
  satellite network with hopping beams,'' \emph{IEEE Commun. Lett.}, vol.~4,
  no.~12, pp. 390--393, 2000.

\bibitem{seyedi:commlett:12}
Y.~Seyedi and S.~M. Safavi, ``On the analysis of random coverage time in mobile
  {LEO} satellite communications,'' \emph{IEEE Commun. Lett.}, vol.~16, no.~5,
  pp. 612--615, 2012.

\bibitem{okati:tcom:20}
N.~Okati, T.~Riihonen, D.~Korpi, I.~Angervuori, and R.~Wichman, ``Downlink
  coverage and rate analysis of low {Earth} orbit satellite constellations
  using stochastic geometry,'' \emph{IEEE Trans. Commun.}, vol.~68, no.~8, pp.
  5120--5134, 2020.

\bibitem{talgat:commlett:20}
A.~Talgat, M.~A. Kishk, and M.-S. Alouini, ``Nearest neighbor and contact
  distance distribution for {Binomial} point process on spherical surfaces,''
  \emph{IEEE Commun. Lett.}, vol.~24, no.~12, pp. 2659--2663, 2020.

\bibitem{talgat:commlett:21}
------, ``Stochastic geometry-based analysis of {LEO} satellite communication
  systems,'' \emph{IEEE Commun. Lett.}, vol.~25, no.~8, pp. 2458--2462, 2021.

\bibitem{hourani:wcl:21}
A.~Al-Hourani, ``An analytic approach for modeling the coverage performance of
  dense satellite networks,'' \emph{IEEE Wireless Commun. Lett.}, vol.~10,
  no.~4, pp. 897--901, 2021.

\bibitem{hourani:wcl:21_2}
------, ``Optimal satellite constellation altitude for maximal coverage,''
  \emph{IEEE Wireless Commun. Lett.}, vol.~10, no.~7, pp. 1444--1448, 2021.

\bibitem{okati:tcom:22}
N.~Okati and T.~Riihonen, ``Nonhomogeneous stochastic geometry analysis of
  massive leo communication constellations,'' \emph{IEEE Trans. Commun.},
  vol.~70, no.~3, pp. 1848--1860, 2022.

\bibitem{renzo:tcom:13}
M.~D. Renzo, A.~Guidotti, and G.~E. Corazza, ``Average rate of downlink
  heterogeneous cellular networks over generalized fading channels: {A}
  stochastic geometry approach,'' \emph{IEEE Trans. Commun.}, vol.~61, no.~7,
  pp. 3050--3071, 2013.

\bibitem{afshang:twc:17}
M.~Afshang and H.~S. Dhillon, ``Fundamentals of modeling finite wireless
  networks using {Binomial} point process,'' \emph{IEEE Trans. Wireless
  Commun.}, vol.~16, no.~5, pp. 3355--3370, 2017.

\bibitem{haenggi:tit:05}
M.~Haenggi, ``On distances in uniformly random networks,'' \emph{IEEE Trans.
  Inform. Theory}, vol.~51, no.~10, pp. 3584--3586, 2005.

\bibitem{cundy:math:89}
H.~Cundy and A.~Rollett, ``Sphere and cylinder----{Archimedes}’ theorem,''
  \emph{Mathematical Models}, pp. 172--173, 1989.

\bibitem{koretz:tsp:09}
A.~Koretz and B.~Rafaely, ``Dolph–chebyshev beampattern design for spherical
  arrays,'' \emph{IEEE Trans. Signal Process.}, vol.~57, no.~6, pp. 2417--2420,
  2009.

\bibitem{giunta:wcl:18}
G.~Giunta, C.~Hao, and D.~Orlando, ``Estimation of {Rician} {K}-factor in the
  presence of nakagami-$m$ shadowing for the {LoS} component,'' \emph{IEEE
  Wireless Commun. Lett.}, vol.~7, no.~4, pp. 550--553, 2018.

\bibitem{singh:twc:13}
S.~Singh, H.~S. Dhillon, and J.~G. Andrews, ``Offloading in heterogeneous
  networks: {Modeling}, analysis, and design insights,'' \emph{IEEE Trans.
  Wireless Commun.}, vol.~12, no.~5, pp. 2484--2497, 2013.

\bibitem{lee:arxiv:22}
\BIBentryALTinterwordspacing
J.~Lee, S.~Noh, S.~Jeong, and N.~Lee, ``Coverage analysis of {LEO} satellite
  downlink networks: Orbit geometry dependent approach,'' \emph{ArXiv}, 2022.
  [Online]. Available: \url{https://arxiv.org/abs/2206.09382}
\BIBentrySTDinterwordspacing

\bibitem{jia:iotj:22}
H.~Jia, Z.~Ni, C.~Jiang, L.~Kuang, and J.~Lu, ``Uplink interference and
  performance analysis for megasatellite constellation,'' \emph{IEEE Internet
  of Things Jour.}, vol.~9, no.~6, pp. 4318--4329, 2022.

\bibitem{alzer:97}
H.~Alzer, ``On some inequalities for the incomplete {Gamma} function,''
  \emph{Mathematics of Computation}, vol.~66, no. 218, pp. 771--778, 1997.

\end{thebibliography}

\end{document}